\documentclass[12pt, journal,onecolumn]{IEEEtran}

\usepackage{setspace}
\usepackage[margin=1in,paperheight=11in,paperwidth=8.5in]{geometry}

\usepackage[utf8]{inputenc}
\usepackage[english]{babel}
\usepackage{parskip}

\usepackage{hyperref} 
\usepackage{subcaption}
\usepackage{graphicx}

\usepackage{comment}
\usepackage[noadjust]{cite}

\usepackage{color}
\usepackage{dsfont}
\usepackage{mleftright}
\usepackage{mdframed}

\usepackage{mathtools}
\usepackage{mathrsfs}
\usepackage{amsfonts,amsthm,amssymb}
\usepackage{bm}

\usepackage[ruled,vlined]{algorithm2e}
\usepackage{listings}

\usepackage{booktabs}

\newtheorem{proposition}{Proposition}
\newtheorem{lemma}{Lemma}

\newtheorem{theorem}{Theorem}
\newtheorem{corollary}{Corollary}
\newtheorem{remark}{Remark}
\newtheorem{construction}{Construction}

\newtheorem{alg}{Algorithm}

\theoremstyle{definition}
\newtheorem{definition}{Definition}

\def\mbA{\mathbf{A}}
\def\mbB{\mathbf{B}}
\def\mbC{\mathbf{C}}

\def\mbE{\mathbf{E}}
\def\mbI{\mathbf{I}}
\def\mbP{\mathbf{P}}
\def\mbU{\mathbf{U}}
\def\mbV{\mathbf{V}}
\def\mbX{\mathbf{X}}
\def\mbY{\mathbf{Y}}
\def\mbZ{\mathbf{Z}}

\def\mb0{\mathbf{0}}
\def\mba{\mathbf{a}}
\def\mbb{\mathbf{b}}

\def\mbd{\mathbf{d}}

\def\mbf{\mathbf{f}}
\def\mbg{\mathbf{g}}

\def\mbp{\mathbf{p}}
\def\mbr{\mathbf{r}}
\def\mbu{\mathbf{u}}
\def\mbn{\mathbf{n}}
\def\mbv{\mathbf{v}}

\def\mbx{\mathbf{x}}
\def\mby{\mathbf{y}}
\def\mbz{\mathbf{z}}

\def\calA{\mathcal{A}}
\def\calB{\mathcal{B}}

\def \calP{\mathcal{P}}
\def \calS{\mathcal{S}}

\def\mbcalA{\bm{\mathcal{A}}}
\def\mbcalB{\bm{\mathcal{B}}}

\def \bbR{\mathbb{R}}

\def\mbalpha{\bm{\alpha}}
\def\mbbeta{\bm{\beta}}
\def\mblambda{\bm{\lambda}}

\def\mbalphaj{\bm{\alpha}^{(j)}}
\def\mbbetai{\bm{\beta}^{(i)}}
\def\mbbetaj{\bm{\beta}^{(j)}}

\DeclareMathOperator{\Tr}{Tr}

\DeclareMathOperator*{\argmax}{\arg\!\max}

\newcommand{\vander}[1]{\mathsf{Vander}(#1)}
\newcommand{\norm}[1]{\left\lVert#1\right\rVert}

\begin{document}
%
\title{$\epsilon$-Approximate Coded Matrix Multiplication is Nearly Twice as Efficient as Exact Multiplication}
%
%
%

\author{Haewon Jeong, Ateet Devulapalli, Viveck R. Cadambe, Flavio Calmon$^{*}$\thanks{$^{*}$Haewon Jeong and Flavio Calmon are with the John A. Paulson School of Engineering and Applied Sciences at Harvard University.  E-mails: \textsf{haewon, flavio@seas.harvard.edu}. Ateet Devulapalli and Viveck R. Cadambe are with the School of Electrical Engineering and Computer Science at Pennsylvania State University.  E-mails: \textsf{azd565,viveck@psu.edu}}} 
\maketitle

\begin{abstract}
We study coded distributed matrix multiplication from an approximate recovery viewpoint. We consider a system of $P$ computation nodes where each node stores $1/m$ of each multiplicand via linear encoding. Our main result shows that the matrix product can be recovered with $\epsilon$ relative error from any $m$ of the $P$ nodes for any $\epsilon > 0$. We obtain this result through a careful specialization of MatDot codes---a class of matrix multiplication codes previously developed in the context of exact recovery ($\epsilon=0$). Since prior results showed that  MatDot codes achieve the best exact recovery threshold for a class of linear coding schemes, our result shows that allowing for mild approximations leads to a system that is nearly twice as efficient as exact reconstruction. Moreover, we  develop an optimization framework based on alternating minimization that enables the discovery of new codes for approximate matrix multiplication.
\end{abstract}

\IEEEpeerreviewmaketitle

\section{Introduction}
Coded computing has emerged as a promising paradigm to resolving straggler and security bottlenecks in large-scale distributed computing platforms~\cite{lee2017coded,yu2017polynomial,MatDotITTrans,dutta2016short,yu2019lagrange,GC2,GC3, halbawi2018improving,reisizadeh2017coded,jeong2018masterless,yu2017coded,Virtualization,ferdinand2016anytime,ferdinand2018hierarchical,mallick2018rateless,wang2018coded,quangQR,severinson2018block,crossiter, gupta2018oversketch,gupta2019oversketched,NewsletterPaper,li2020coded,dutta2020addressing}. The foundations of this paradigm lie in novel code constructions for elemental computations such as matrix operations and polynomial computations, and fundamental limits on their performance. In this paper, we show that the state-of-the-art fundamental limits for such elemental computations grossly underestimate the performance by focusing on \emph{exact} recovery of the 
computation output. By allowing for mild \emph{approximations} of the computation output, we demonstrate significant improvements in terms of the trade-off between fault-tolerance and the degree of redundancy.

Consider a distributed computing system with $P$ nodes for performing the matrix multiplication $\mathbf{A}\mathbf{B}.$ If each node is required to store a fraction $1/m$ of both matrices, the best known recovery threshold is equal to $2m-1$ achieved by the  MatDot code \cite{MatDotITTrans}. Observe the contrast between distributed coded \emph{computation} with distributed data \emph{storage}, where a maximum distance separable (MDS) code  ensures that if each node stores a fraction $1/m$ of the data, then the data can be recovered from any $m$ nodes\footnote{This essentially translates to the Singleton bound being tight for a sufficiently large alphabet} \cite{roth2006introduction}. Indeed, the recovery threshold of $m$ is crucial to the existence of practical codes  that bring  fault-tolerance to large-scale data storage systems with relatively minimal overheads (e.g., single parity and Reed-Solomon codes \cite{Balaji2018ErasureOverview}). 

The contrast between data storage and computation is even more pronounced when we consider the generalization of matrix-multiplication towards multi-variate polynomial evaluation $f(\mathbf{A}_{1},\mathbf{A}_{2},\ldots,\mathbf{A}_{\ell})$ where each node is allowed to store a fraction $1/m$ of each of $\mathbf{A}_1,\mathbf{A}_2,\ldots,\mathbf{A}_\ell$. In this case, the technique of Lagrange coded-computing \cite{yu2019lagrange} demonstrates that the recovery threshold is $d(m-1)+1,$ where $d$ is the degree of the polynomial. Note that a recovery threshold of $m$ is only obtained for the special case of degree $d=1$ polynomials, i.e., elementary linear transformations that were originally studied in \cite{lee2017speeding}. While the results of \cite{MatDotITTrans, yu2020transIT} demonstrate that the amount of redundancy is much less than previously thought for degree $d > 1$ computations, these codes still require an overwhelming amount of additional redundancy---even to tolerate a single failed node---when compared to codes for distributed storage.

\subsection{Summary of Results}

Our paper is the result of the search for an analog of MDS codes
---in terms of the amount of redundancy required---for  coded-computation of polynomials with degree greater than $1$. We focus on the case of coded matrix multiplication where the goal is to recover the matrix product  $\mathbf{C}=\mathbf{A}\mathbf{B}$. We consider a distributed computation system of $P$ worker nodes  similar to \cite{yu2017polynomial,MatDotITTrans}; we allow each worker to store an $m$-th fraction of  matrices of $\mathbf{A},$ $\mathbf{B}$ via linear transformations (encoding). The workers output the product of the encoded matrices. A central master/fusion node collects the output of a set $\calS$ of non-straggling workers and aims to decode $\mathbf{C}$ with a relative error of $\epsilon$. The recovery threshold $K(m,\epsilon)$ is the cardinality of the largest minimal subset $\calS$ that allows for such recovery. It has been shown in \cite{MatDotITTrans,yu2020transIT} that, for natural classes of linear encoding schemes, $K(m,0) = 2m-1.$ 

Our main result shows that the MatDot code with a specific set of evaluation points is   able to achieve $K(m,\epsilon) = m$, remarkably, for \emph{any} $\epsilon > 0.$ A simple converse shows that the our result is tight for $0 < \epsilon < 1$ for unit norm matrices. Our results mirrors several results in classical information theory (e.g., almost lossless data compression), where allowing $\epsilon$-error for any $\epsilon > 0$ 
 leads to surprisingly significant improvements in performance. We also show that for PolyDot/Entangled polynomial codes \cite{MatDotITTrans,yu2020transIT,dutta2018unified} where matrices $\mathbf{A},\mathbf{B}$ are restricted to be split as $p\times q$ and $q \times p$ block matrices respectively, we improve the recovery threshold\footnote{Strictly speaking, the recovery threshold of entangled polynomial codes depends on the bilinear complexity, which can be smaller than $p^2q+q-1$ \cite{yu2020transIT}.} from $p^2 q + q -1$  to $p^2 q$ by allowing  $\epsilon$-error.
 We believe that these results open up a new avenue in coded computing research via revisiting existing code constructions and allowing for an $\epsilon$-error.  
  
A second contribution of our paper is the development of an optimization formulation that enables the discovery of new coding schemes for approximate computing.  We show that the optimization can be solved through an alternating minimization algorithm that has simple, closed-form iterations as well as provable convergence to a local minimum. We illustrate through numerical examples that our optimization approach finds approximate computing codes with favourable trade-offs between approximation error and recovery threshold. Through an application of our code constructions to distributed training for classification via logistic regression, we show that our approximations suffices to obtain accurate classification results in practice. 

\subsection{Related Work}
The study of coded computing for elementary linear algebra operations, starting from \cite{lee2017speeding, dutta2016short}, is an active research area (see surveys \cite{NewsletterPaper,li2020coded,dutta2020addressing}). Notably, the recovery thresholds for matrix multiplication were established via achievability and converse results respectively in \cite{yu2017polynomial, MatDotITTrans, yu2020transIT}. The Lagrange coded computing framework of \cite{yu2019lagrange} generalized the systematic MatDot code construction of \cite{MatDotITTrans} to the context of multi-variate polynomial evaluations and established a tight lower bound on the recovery threshold. These works focused on exact recovery of the computation output.

References \cite{wang2019fundamental,charles2017approximate,9081964} studied the idea of gradient coding from an approximation viewpoint, and demonstrated improvements in recovery threshold over exact recovery. However, in contrast with our results, the error obtained either did not correct all possible error patterns with a given recovery threshold (i.e., they considered a probabilistic erasure model), and the relative error of their approximation was lower bounded. The references that are most relevant to our work are \cite{gupta2018oversketch,gupta2019oversketched,jahani2019codedsketch}, which also aim to improve the recovery threshold of coded matrix multiplication by allowing for a relative error of $\epsilon.$ These references use random linear coding (i.e., sketching) techniques to obtain a recovery threshold $\overline{K}(\epsilon,\delta,m)$ where $\delta$ is the probability of failing to recover the matrix product with a relative error of $\epsilon$; the problem statement of \cite{jahani2019codedsketch} is particularly similar to ours. Our results can be viewed as a strict improvement over this prior work, as we are able to obtain a recovery threshold of $m$ even with $\delta=0$, whereas the recovery threshold is at least $2m-1$ for $\delta=0$ in \cite{gupta2018oversketch,gupta2019oversketched,jahani2019codedsketch}.

A related line of work in \cite{jahani2020berrut,soleymani2020analog}  study coded polynomial evaluation beyond exact recovery and note techniques to improve the quality of the approximation. References \cite{kosaian2019parity, kosaian2020learning} develops machine learning techniques for approximate learning; while they show empirical existence codes with low recovery thresholds (such as single parity codes \cite{kosaian2019parity}) for learning tasks they do not provide theoretical guarantees. Specifically, while \cite{kosaian2019parity, kosaian2020learning} shows the benefits of approximation in terms of recovery threshold, it is unclear whether these benefits appear in their scheme due to the special structure of the data, or whether the developed codes work for all realizations of the data. In contrast with \cite{kosaian2019parity,kosaian2020learning,  jahani2020berrut, soleymani2020analog}, we are the first to establish the strict gap in the recovery thresholds for  $\epsilon$-error computations versus exact computation for matrix multiplication, which is a canonical case of degree $2$ polynomial evaluation.
 
A tangentially related body of work \cite{fahim2019numerically, ramamoorthy2019numerically, subramaniam2019random,charalambides2020numerically} studies the development of numerically stable coded computing techniques. While some of these works draw on techniques from approximation theory, they focus on maintaining recovery threshold the same as earlier constructions, but bounding the approximation error of the output in terms of the precision of the computation.

\section{System Model and Problem Statement} \label{sec:model}
\subsection{Notations} \label{subsec:notations}
We  define $[n]\triangleq \{1, 2, \cdots, n \}$. We use bold fonts for vectors and matrices. $A[i,j]$ denotes the $(i,j)$-th entry of an $M \times N$ matrix $\mbA$ ($i\in [M], j\in [N]$) and $v[i]$ is  the $i$-th entry of a length-$N$ vector $\mbv$  ($i\in [N]$). 

\subsection{System Model} \label{subsec:sys_mod}
 We  consider a distributed computing system with a master node and $P$ worker nodes. At the beginning of the computation, a master node distributes appropriate tasks and inputs to worker nodes. Worker nodes perform the assigned task and send the result back to the master node. Worker nodes are prone to failures or delay (stragglers). Once the master node receives results from a sufficient number of worker nodes, it produces the final output. 

 We are interested in distributed matrix multiplication, where the goal is to compute 
\begin{equation}\label{eq:C_AB}
    \mbC = \mbA \cdot \mbB.
\end{equation}
We assume $\mbA, \mbB \in \bbR^{n\times n}$ are matrices with a bounded norm, i.e., 
\begin{equation}
    \label{eq:norm-constraint}
    ||\mbA||_F \leq \eta \;\; \text{ and } \; ||\mbB||_F \leq \eta, 
\end{equation}
where $||\cdot||_F$ denotes Frobenius norm. We further assume that worker nodes have memory constraints such that each node can hold only an $m$-th fraction of $\mbA$ and an $m$-th fraction of $\mbB$ in  memory. To meet the memory constraint, we divide $\mbA, \mbB$ into small equal-sized sub-blocks as follows\footnote{We limit ourselves to splitting the input matrices into a grid of submatrices. Splitting into an arbitrary shape is beyond the scope of this work.}: 
\small
\begin{equation} 
    \mbA = \begin{bmatrix} 
    \mbA_{1,1} & \cdots & \mbA_{1,q} \\ 
    \vdots & \ddots & \vdots \\ 
     \mbA_{p,1} & \cdots & \mbA_{p,q}
    \end{bmatrix}, 
    \mbB = \begin{bmatrix} 
    \mbB_{1,1} & \cdots & \mbB_{1,p} \\ 
    \vdots & \ddots & \vdots \\ 
     \mbB_{q,1} & \cdots & \mbB_{q,p}
    \end{bmatrix},
    \label{eq:ABmodel}
\end{equation}
\normalsize
where $pq = m$. When $p=1,$ we simply denote $$\mathbf{A}=\begin{bmatrix}\mathbf{A}_{1} & \mathbf{A}_{2} & \ldots & \mathbf{A}_{m}\end{bmatrix} \mbox{ and } \mathbf{B}=\begin{bmatrix}\mathbf{B}_{1} \\ \mathbf{B}_{2} \\ \vdots \\ \mathbf{B}_{m}\end{bmatrix}.$$
    
To mitigate failures or stragglers, a master node encodes redundancies through linear encoding. The $i$-th worker node receives encoded inputs $\widetilde{\mbA}_i$ and $\widetilde{\mbB}_i$ such that:
    \begin{align*}
        \widetilde{\mbA}_i = f_i(\mbA_{1,1}, \cdots, \mbA_{p,q}), \;\;
        \widetilde{\mbB}_i = g_i(\mbB_{1,1}, \cdots, \mbB_{q,p}),
    \end{align*}
    where 
    \begin{align}
        f_i: \underbrace{\bbR^{\frac{n}{p}\times\frac{n}{q}} \times \cdots \times \bbR^{\frac{n}{p}\times\frac{n}{q}}}_{pq=m} \rightarrow \bbR^{\frac{n}{p}\times\frac{n}{q}},
        \label{eq:fimapdef}\\ 
        g_i: \underbrace{\bbR^{\frac{n}{q}\times\frac{n}{p}} \times \cdots \times \bbR^{\frac{n}{q}\times\frac{n}{p}}}_{m} \rightarrow \bbR^{\frac{n}{q}\times\frac{n}{p}}.
        \label{eq:gimapdef}
    \end{align}
    
    We assume that $f_i,g_i$ are linear, i.e., their outputs are linear combinations of  $m$ inputs. For example, we may have
    $
        f_i(\mbZ_1, \cdots, \mbZ_m) = \gamma_{i,1} \mbZ_1 + \cdots + \gamma_{i,m} \mbZ_m
   $
    for some $\gamma_{i,j} \in \bbR \; (j\in [m])$. .
    
Worker nodes are oblivious of the encoding/decoding process and simply perform matrix multiplication on the inputs they receive. In our case, each worker node computes  
    \begin{equation}
        \widetilde{\mbC}_i = \widetilde{\mbA}_i \cdot \widetilde{\mbB}_i, 
    \end{equation}
    and  returns the $\frac{n}{p}\times\frac{n}{p}$ output matrix $\widetilde{\mbC}_i$  to the master node. 
    
Finally, when the master node receives outputs from a subset of worker nodes, say  $\calS \subseteq [P]$, it performs decoding:
    \begin{equation}
        \widehat{\mbC}_{\calS} = d_{\calS}((\widetilde{\mbC}_i)_{i \in \calS}), 
    \end{equation}
    where $\left\{d_{\calS}\right\}_{\calS\subseteq [P]}$ is a set of predefined decoding functions that take $|\calS|$ inputs from $\bbR^{\frac{n}{p} \times \frac{n}{p}}$ and outputs an $n$-by-$n$ matrix. Note that we do not restrict the decoders $d_{\calS}$ to be linear.

\subsection{Approximate Recovery Threshold} \label{subsec:approx_rt}
Let $\mbf$ and $\mbg$ be vectors of linear encoding functions: $$\mbf = \begin{bmatrix} f_1 & \cdots & f_{P} \end{bmatrix}, \mbg = \begin{bmatrix} g_1 & \cdots & g_{P} \end{bmatrix},$$
and let $\mbd$ be a length-$2^P$ vector of decoding functions $d_{\calS}$ for all subsets $\calS \subseteq [P]$. More specifically, $d_{\calS}$ is a decoding function for the scenario where worker nodes in set $\calS$ are successful in returning their computations to the master node and all other worker nodes fail. We say that the $\epsilon$-approximate recovery threshold of $\mbf, \mbg, \mbd$ is $K$ if  for any $\mbA$ and $\mbB$ that satisfy the norm constraints \eqref{eq:norm-constraint}, the decoded matrix satisfies

    \begin{equation}
    \label{eq:epsilondef}
        |\widehat{C}_{\calS}[i,j] - C[i,j]| \leq \epsilon \quad (i,j \in [n])
    \end{equation}
    for every $\calS \subseteq [P]$ such that $|\calS| \geq K$. We denote this recovery threshold as $K(m, \epsilon, \mbf, \mbg, \mbd)$. Moreover, let $K^{*}(m, \epsilon)$ be defined as the minimum of $K(m, \epsilon, \mbf, \mbg, \mbd)$ over all possible linear functions $\mbf$, $\mbg$ and all possible decoding functions $\mbd$, i.e.,
    \begin{equation}\label{eq:K_approx_def}
        K^{*}(m, \epsilon) \triangleq \min_{\mbf, \mbg, \mbd}  K(m, \epsilon, \mbf, \mbg, \mbd). 
    \end{equation}
    Note that parameters $p$ and $q$ are embedded in $\mbf$ and $\mbg$ and hence $K^{*}(m, \epsilon)$ is the minimum over all combinations of $p, q$ such that $pq = m$.
Through an achievability scheme in \cite{MatDotITTrans} and a converse in \cite{yu2020transIT}, for exact recovery, the optimal threshold has been characterized to be $2m-1$: 
   \begin{theorem} [Adaptation of Theorem 2 in \cite{yu2020transIT} and Theorem III.1 in \cite{MatDotITTrans}]
   Under the system model given in Section~\ref{subsec:sys_mod}
    \begin{equation}
        K^{*}(m,\epsilon=0) = 2m -1.
    \end{equation}
   \end{theorem}
    
\subsection{Summary of Main Result}
Our main result is summarized the following theorem:

\begin{mdframed}
\begin{theorem}
Under the system model given in Section~\ref{subsec:sys_mod}, the optimal $\epsilon$-approximate recovery threshold is:
\begin{equation}
    K^{*}(m, \epsilon) = m. 
\end{equation}
\textbf{(Achievability -- Theorem~\ref{thm:approx_matdot})}

For any $0 < \epsilon < \min (2, 3\eta^2 \sqrt{2m-1})$, the $\epsilon$-approximate MatDot codes in Construction~\ref{const:approx_matdot} achieves:
\begin{equation*}
    K(m,\epsilon, \mbf_{\epsilon\text{-MatDot}}, \mbg_{\epsilon\text{-MatDot}}, \mbd_{\epsilon\text{-MatDot}}) = m.
\end{equation*}
\textbf{(Converse -- Theorem~\ref{thm:con})}

For all $0 < \epsilon < \eta^2$, $K^{*}(m, \epsilon) \geq m.$
\end{theorem}
\end{mdframed}

The achievability scheme given in Theorem~\ref{thm:approx_matdot} is  only for $p=1$. For a fixed $p>1$, we propose $\epsilon$-approximate PolyDot strategy which reduces the recovery threshold of Entangled-Poly codes from $p^2q+q-1$ to $p^2q$ by allowing $\epsilon$-error in the recovered output (Theorem~\ref{thm:approx_polydot}).

\section{Theoretical Characterization of $K^{*}(m,\epsilon)$} \label{sec:approx_matdot}
In this section, we first propose the construction of $\epsilon$-approximate MatDot codes that can achieve the recovery threshold of $m$ for $\epsilon$ approximation error. Then, we prove the converse result which states that the recovery threshold cannot be smaller than $m$ for sufficiently small $\epsilon$. 

\subsection{Approximate MatDot Codes} \label{subsec:approx_matdot}
We briefly introduce the construction of MatDot codes and then  show that a simple adaptation of MatDot codes can be used for approximate coded computing. 

\begin{construction}[MatDot Codes~\cite{MatDotITTrans}] \label{const:MatDot}
Define polynomials $p_{\mbA}(x)$ and $p_{\mbB}(x)$ as follows: 
\begin{equation} \label{eq:MatDot_pA_pB}
    p_\mathbf{A}(x)=\sum_{i=1}^{m} \mathbf{A}_i x^{i-1}, p_\mathbf{B}(x)=\sum_{j=1}^{m} \mathbf{B}_j x^{m-j}.
\end{equation}
Let $\lambda_1, \lambda_2, \ldots, \lambda_P$ be $P$ distinct elements in $\bbR$. The $i$-th worker receives 
encoded versions of matrices:
\begin{align*}
    \widetilde{\mbA}_i &= p_{\mbA}(\lambda_i) = \mbA_1 + \lambda_i \mbA_2 + \cdots + \lambda_i^{m-1} \mbA_{m}, \\
    \widetilde{\mbB}_i &= p_{\mbB}(\lambda_i) = \mbB_m + \lambda_i \mbB_{m-1} + \cdots + \lambda_i^{m-1} \mbB_{1},
\end{align*}
and then computes matrix multiplication on the encoded matrices: 
\begin{equation*}
    \widetilde{\mbC}_i = \widetilde{\mbA}_i \widetilde{\mbB}_i = p_{\mbA}(\lambda_i) p_{\mbB}(\lambda_i) = p_{\mbC} (\lambda_i).
\end{equation*}
The polynomial $p_{\mbC}(x)$ has degree $2m-2$ and has the following form: 
\begin{equation}
p_{\mathbf{C}}(x) = \sum_{i=1}^{m} \sum_{j=1}^{m} \mathbf{A}_i \mathbf{B}_j x^{m-1+(i-j)}.
\end{equation}
Once the master node receives outputs from $2m-1$ successful worker nodes, it can recover the coefficients of  $p_{\mbC}(x)$ through polynomial interpolation, and then  recover $\mbC=\sum_{i=1}^{m} \mathbf{A}_{i}\mathbf{B}_{i}$ as the coefficient of $x^{m-1}$ in $p_{\mathbf{C}}(x)$.
\hfill $\square$
\end{construction}

The recovery threshold of MatDot codes is $2m-1$ because the output polynomial $p_{\mbC}(x)$ is a degree-$(2m-2)$ polynomial and we need $2m-1$ points to recover all of the coefficients of $p_{\mbC}(x)$. However, in order to recover $\mbC$, we only need the coefficient of $x^{m-1}$ in $p_{\mbC}(x)$. The key idea of Approximate MatDot Codes is to carefully choose the evaluation points that reduce this overhead. In fact, we select evaluation points in a small interval that is proportional to $\epsilon$. 

\begin{construction}[$\epsilon$-Approximate MatDot codes] \label{const:approx_matdot}
Let $\mbA$ and $\mbB$ be matrices in $\bbR^{n\times n}$ that satisfy $||\mbA||_F, ||\mbB||_F \leq \eta$. Let $\epsilon \in \bbR$ be a constant such that
\begin{equation}
    0 < \epsilon < \min(2, 3\eta^2\sqrt{2m-1}).
\end{equation}
Then, $\epsilon$-Approximate MatDot code is a MatDot code defined in Construction~\ref{const:MatDot} with evaluation points $\lambda_1, \ldots, \lambda_P$ that satisfy:
\begin{equation}\label{eq:appox_matdot_cond}
    |\lambda_i| < \frac{\epsilon}{6 \eta^2 \sqrt{2m-1}(m-1)m}, \;\; i\ \in [P].
\end{equation}
\end{construction}

We  then show that this construction has the approximate recovery threshold of $m$. 
\begin{theorem}\label{thm:approx_matdot}
For any $0 < \epsilon < \min (2, 3\eta^2 \sqrt{2m-1})$, the $\epsilon$-Approximate MatDot codes in Construction~\ref{const:approx_matdot} achieves:
\begin{equation} \label{eq:app_matdot_thm}
    K(m,\epsilon, \mbf_{\epsilon\text{-MatDot}}, \mbg_{\epsilon\text{-MatDot}}, \mbd_{\epsilon\text{-MatDot}}) = m,
\end{equation}
where $\mbf_{\epsilon\text{-MatDot}}, \mbg_{\epsilon\text{-MatDot}}, \mbd_{\epsilon\text{-MatDot}}$ are encoding and decoding functions specified by Construction~\ref{const:approx_matdot}.
\end{theorem}

\begin{remark}
When $\epsilon \geq \min (2, 3\eta^2 \sqrt{2m-1})$, we can use $\epsilon'$-Approximate MatDot codes for some $\epsilon' < \min (2, 3\eta^2 \sqrt{2m-1})$. Then, \eqref{eq:app_matdot_thm} can be expressed as:
\begin{equation}
    K(m,\epsilon, \mbf_{\epsilon'\text{-MatDot}}, \mbg_{\epsilon'\text{-MatDot}}, \mbd_{\epsilon'\text{-MatDot}}) = m.
\end{equation}
\end{remark}

\begin{remark}
The error bound provided by Theorem~\ref{thm:approx_matdot} is an absolute bound, i.e., $||\widehat{\mbC} - \mbC||_{\max} \leq \epsilon$. This is because we choose the evaluation points which are scaled by $\eta$, which is the upper bound of $||\mbA||_F$ and $||\mbB||_F$ as given in \eqref{eq:appox_matdot_cond}. If we do not assume prior knowledge on the upper bound of $||\mbA||_F$ and $||\mbB||_F$, we can choose $\lambda_i$'s to be some small numbers, e.g., $|\lambda_i| \leq \Delta$, and then the error bound will be a relative bound on $\frac{||\widehat{\mbC} - \mbC||_{\max}}{||\mbA||_F  ||\mbB||_F}$. Furthermore, note that the bound on the max norm can be easily converted to bounds on other types of norm (e.g., Frobenius norm or 2-norm) within a constant factor using matrix norm equivalence relations. 
\end{remark}
While we defer the full proof to Appendix~\ref{app:approx_matdot}, we provide an intuitive explanation of the above theorem. 
\subsection{An insight behind Approximate MatDot Codes}

 Let $S(x)$ be a polynomial of degree $2m-1$ and let $P(x)$ be a polynomial of degree $m$. Then, $S(x)$ can be written as:
\begin{equation}
    S(x) = P(x) Q(x) + R(x),
\end{equation}
and the degree of $Q$ and $R$ are both at most $m-1$. Now, let $\lambda_1, \ldots, \lambda_m$ be the roots of $P(x)$. Then, 
\begin{equation}
    S(\lambda_i) = R(\lambda_i). 
\end{equation}
If we have $m$ evaluations at these points, we can exactly recover the coefficients of the polynomial $R(x)$. 

Recall that we only need the coefficient of $x^{m-1}$ in MatDot codes. Letting $P(x)=x^m$, $S(x)$ can be written as:
\begin{equation}
    S(x) = x^m Q(x) + R(x).
\end{equation}
Since the lower order terms are all in $R(x)$, the coefficient of $x^{m-1}$ in $R(x)$ is equal to the coefficient of $x^{m-1}$ in $S(x)$. Thus, recovering the coefficients of $R(x)$ is sufficient for MatDot decoding. However, $x^m$ has only one root, $0$. For approximate decoding, we can use points close to $0$ as evaluation points to make $x^m \approx 0$. Then, we have:
\begin{equation}
    S(\lambda_i) = \lambda_i^m Q(\lambda_i) + R(\lambda_i) \approx R(\lambda_i). 
\end{equation}
When $|S(\lambda_i)-R(\lambda_i)|$ is small, we can use $m$ evaluations of $S(\lambda_i)$'s to approximately interpolate $R(x)$. Moreover, when $\lambda_i$ is small, we can also bound $|S(\lambda_i)-R(\lambda_i)|$ when $Q$ has a bounded norm. In our case, $S$ has a bounded norm due the norm constraints \eqref{eq:norm-constraint} on the input matrices and, thus, $Q$ must have a bounded norm since the higher-order terms in $S$ are solely determined by $Q$. 

\subsection{Converse}

We have shown that for \textit{any} matrices $\mbA$ and $\mbB$, and with a recovery threshold of $m$,  $\epsilon$-approximate MatDot codes  can achieve arbitrarily small error. We now show a converse indicating that for a recovery threshold of $m-1$, there exists matrices $\mbA\in \mathbb{R}^{n\times n}$ and $\mbB\in \mathbb{R}^{n\times n}$ where the error cannot be made arbitrarily small for any type of encoding.

\begin{theorem}
\label{thm:con}
	Under the system model given in Section~\ref{sec:model}, for any  $0< \epsilon < \eta^2$, 
\begin{equation}
    K^{*}(m,\epsilon) \geq m.
\end{equation}
\end{theorem}
Proof is given in Appendix~\ref{app:conproof}.

\subsection{Approximate PolyDot Codes} 
The construction of $\epsilon$-approximate MatDot codes achieves the optimal $\epsilon$- approximate recovery threshold, but is limited to $p=1, q=m$ in \eqref{eq:ABmodel}. For arbitrary $p$ and $q$, the recovery threshold of $p^2q + q -1$ is achieved by PolyDot codes (Entangled-Poly codes). In this section, we show that---similarly to MatDot Codes---the recovery threshold of Polydot codes can be improved  by allowing an $\epsilon$-approximation of the matrix multiplication and selecting evaluation points near zero. We briefly review next the construction of PolyDot codes \cite{MatDotITTrans} (also known as Entangled-Poly codes \cite{yu2020transIT}). 

\begin{construction} [PolyDot (Entangled-Poly) Codes \cite{MatDotITTrans,yu2020transIT}] \label{const:polydot}
In \cite{MatDotITTrans}, a general framework for PolyDot codes is proposed as follows:
\begin{equation}
    p_{\mbA}(x,y) = \sum_{i=1}^{p} \sum_{j=1}^{q} \mbA_{i,j} x^{i-1} y^{j-1}, \;\; p_{\mbB}(y,z) = \sum_{k=1}^{q} \sum_{l=1}^{p} \mbB_{k,l} y^{q-k} z^{l-1},
\end{equation}
where the input matrices  are split as \eqref{eq:ABmodel}. Substituting $x = y^q$ and $z =y^{pq}$ results in Entangled-Poly codes~\cite{yu2020transIT}:
\begin{equation}
    p_{\mbA}(y) = \sum_{i=1}^{p} \sum_{j=1}^{q} \mbA_{i,j} y^{q(i-1) + (j-1)}, \;\; 
    p_{\mbB}(y) = \sum_{k=1}^{q} \sum_{l=1}^{p} \mbB_{k,l} y^{q-k + pq(l-1)}.
\end{equation}
In the product polynomial $p_{\mbC}(y) = p_{\mbA}(y) p_{\mbB}(y) $, the coefficient of $y^{(i-1)q + q-1 + pq(l-1)} = y^{iq +pq(l-1) -1}$ is $\mbC_{i,l} = \sum_{k=1}^{q} \mbA_{i,k} \mbB_{k,l}$. The degree of $p_{\mbC}$ is:
$$(p-1)q +(q-1) + (q-1) + pq(p-1) = p^2q +q -2.$$
Hence, the recovery threshold of $p^2q +q -1$ is achievable.
\end{construction}

We describe next a construction for $\epsilon$-approximate PolyDot codes. 
\begin{construction}[$\epsilon$-Approximate PolyDot codes] \label{const:approx_polydot}
Let $\mbA$ and $\mbB$ be matrices in $\bbR^{n\times n}$ that satisfy $||\mbA||_F, ||\mbB||_F \leq \eta$ and let $\epsilon > 0$ be a constant. 
Then, $\epsilon$-Approximate PolyDot code is a PolyDot code defined in Construction~\ref{const:polydot} with evaluation points $\lambda_1, \ldots, \lambda_P$ that satisfy:
\begin{equation}\label{eq:appox_polydot_cond}
    |\lambda_i| < \min\left( \frac{\epsilon}{\eta^2 q (p^2q-1)}, \frac{1}{p^2q-1}\right), \;\; i\ \in [P].
\end{equation}
\end{construction}

The following theorem states that the recovery threshold can be reduced by $q-1$ by allowing  $\epsilon$-approximate recovery. 
\begin{theorem}\label{thm:approx_polydot}
For any $\epsilon>0$, the $\epsilon$-approximate PolyDot codes in Construction~\ref{const:approx_polydot} achieves:
\begin{equation*}
    K(m,\epsilon, \mbf_{\epsilon\text{-PolyDot}}, \mbg_{\epsilon\text{-PolyDot}}, \mbd_{\epsilon\text{-PolyDot}}) = p^2q =pm ,
\end{equation*}
where $\mbf_{\epsilon\text{-PolyDot}}, \mbg_{\epsilon\text{-PolyDot}}, \mbd_{\epsilon\text{-PolyDot}}$ are encoding and decoding functions specified by Construction~\ref{const:approx_polydot}.
\end{theorem}

Notice that PolyDot codes are a generalized version of MatDot codes, i.e., by setting $p=1, q=m$, Construction~\ref{const:polydot} reduces to MatDot codes. Hence, the result in Theorem~\ref{thm:approx_polydot} also applies to $\epsilon$-approximate MatDot codes. In fact, the proof of this theorem yields a slightly improved  error bound of $\epsilon$-approximate MatDot codes given in Section~\ref{subsec:approx_matdot}. 

\begin{remark}
The condition on the evaluation points $\lambda_i$'s in Construction~\ref{const:approx_matdot} can be relaxed to: 
\begin{equation}
    |\lambda_i| < \min\left( \frac{\epsilon}{\eta^2 \cdot m (m-1)}, \frac{1}{m}\right), \;\; i\ \in [P].
\end{equation}
\end{remark}

The techniques used in the proofs of Theorem \ref{thm:approx_matdot} and Theorem \ref{const:approx_polydot} are distinct, yet both proofs are sufficient to demonstrate the recovery threshold of $\epsilon$-approximate MatDot codes. We present both in this paper since they may serve as blueprints for future $\epsilon$-approximate code constructions.

\section{An Optimization Approach to Approximate Coded Computing} \label{sec:opt}

The  Approximate MatDot code construction shows the theoretical possibility that the recovery threshold can be brought down from $2m-1$ to $m$.
 
In this section, we propose another approach to find an approximate coded computing strategy. As we are not aiming for zero error, we pose the question as an optimization problem where the difference between the original matrix and the reconstructed matrix is minimized. The goal of optimization is  to find $\epsilon$ such that $K^{*}(m, \epsilon) \leq k$ for a given $k$, within the space of linear encoding and decoding functions.


In Section~\ref{subsec:opt_example}, we illustrate our optimization framework through a simple setting of $P=3 $ nodes. In Section~\ref{subsec:opt_method}, we describe our formulation formally, for arbitrary values of parameters $P,K,m.$. We report numerical  results for our optimization algorithm in Section~\ref{subsec:opt_result}.
\subsection{A simple example} \label{subsec:opt_example}
Consider an example of $m=2, k=2, P= 3$. The input matrices are split into: 
\begin{equation}
    \mbA = \begin{bmatrix} \mbA_1 & \mbA_2 \end{bmatrix}, \quad \mbB = \begin{bmatrix} \mbB_1 \\ \mbB_2 \end{bmatrix}.
\end{equation}
As $f_i$'s and $g_i$'s are linear encoding functions, let
\small
\begin{equation}
    \mbalpha^{(i)} = \begin{bmatrix}
     \alpha_1^{(i)} \\  \alpha_2^{(i)}
    \end{bmatrix}, \mbbeta^{(i)} = \begin{bmatrix}
     \beta_1^{(i)} \\ \beta_2^{(i)}
    \end{bmatrix}
\end{equation}
\normalsize
be the encoding coefficients for $\mbA$ and $\mbB$ for the $i$-th node. The $i$-th worker node receives encoded inputs:
\begin{equation*}
    \widetilde{\mbA}^{(i)} = \alpha^{(i)}_1 \mbA_1 + \alpha^{(i)}_2 \mbA_2, \;\; \widetilde{\mbB}^{(i)} = \beta^{(i)}_1 \mbB_1 + \beta^{(i)}_2 \mbB_2.
\end{equation*}
The matrix product output at the $i$-th worker node is:
\begin{align*}
     \widetilde{\mbC}^{(i)} = \widetilde{\mbA}^{(i)} \widetilde{\mbB}^{(i)} &=  \alpha^{(i)}_1 \beta^{(i)}_1 \cdot
    \mbA_1 \mbB_1 + 
    \alpha^{(i)}_1 \beta^{(i)}_2 \cdot \mbA_1 \mbB_2\\  &+
    \alpha^{(i)}_2 \beta^{(i)}_1 \cdot \mbA_2\mbB_1 + 
    \alpha^{(i)}_2  \beta^{(i)}_2 \cdot \mbA_2 \mbB_2.
\end{align*}

The recovery threshold $k =2$ implies that with any two $\widetilde{\mbC}^{(i)}$, $\widetilde{\mbC}^{(j)}$, $i\neq j,~i, j\in [3]$, the master node can recover: 
\begin{align*}
    \mbC &= \mbA_1 \mbB_1 + \mbA_2 \mbB_2 = 1 \cdot \mbA_1 \mbB_1 + 0 \cdot \mbA_1 \mbB_2 +  0 \cdot \mbA_2 \mbB_1 + 1 \cdot \mbA_2 \mbB_2.
\end{align*}
For illustration, assume that nodes $i=1$ and $j=2$ responded first. For linear decoding, our goal is to determine decoding coefficients $d_1,d_2\in \bbR$ that yield
\begin{align*}
    \mbC = d_1 \widetilde{\mbC}^{(1)} + d_2 \widetilde{\mbC}^{(2)}.
\end{align*}
For the previous equality to hold for any $\mbA$ and $\mbB$, the coefficients must satisfy: 
\small
\begin{align}\label{eq:ex_enc_dec}
    \begin{bmatrix} 1 &  0 & 0 & 1 \end{bmatrix} &= 
     d_1 \begin{bmatrix} \alpha^{(1)}_1 \beta^{(1)}_1 & \alpha^{(1)}_1 \beta^{(1)}_2 & \alpha^{(1)}_2 \beta^{(1)}_1  & \alpha^{(1)}_2  \beta^{(1)}_2 \end{bmatrix} \nonumber \\ 
     &+ d_2 \begin{bmatrix} \alpha^{(2)}_1 \beta^{(2)}_1 & \alpha^{(2)}_1 \beta^{(2)}_2 & \alpha^{(2)}_2 \beta^{(2)}_1  & \alpha^{(2)}_2  \beta^{(2)}_2 \end{bmatrix} 
\end{align}
\normalsize
By reshaping the length-4 vectors in \eqref{eq:ex_enc_dec} into $2\times 2$ matrices and denoting the identity matrix by $\mbI_{2\times 2}$, \eqref{eq:ex_enc_dec} is equivalent to
\begin{equation} \label{eq:ex_enc_dec2}
    \mbI_{2\times 2} = \sum_{i=1}^{2} d_i \mbalpha^{(i)}\mbbeta^{(i)T}. 
\end{equation}

Encoding coefficients $\mbalpha^{(i)}$'s, $\mbbeta^{(i)}$'s and the decoding coefficients $d_i$'s that satisfy the equality in~\eqref{eq:ex_enc_dec2} would guarantee exact recovery for any input matrices $\mbA$ and $\mbB$. However, we are interested in \emph{approximate} recovery, which means that we want the LHS and RHS in~\eqref{eq:ex_enc_dec2} to be approximately equal. Hence, the goal of optimization is to find encoding and decoding coefficients that minimize the difference between LHS and RHS in~\eqref{eq:ex_enc_dec2}. One possible objective function for this is:
\begin{equation}
    || \mbI_{2\times 2} - \sum_{i=1}^{2} d_i \mbalpha^{(i)}\mbbeta^{(i)T} ||_F^2.
\end{equation}
Recall that this is for the scenario where the third node fails and the first two nodes are successful. There are $\binom{3}{2} = 3$ scenarios where two nodes out of three nodes are successful. For the final objective function, we have to add such loss function for each of these three scenarios. We formalize this next. 

\subsection{Optimization Formulation} \label{subsec:opt_method}
We formulate the optimization framework for arbitrary values of $m$, $k$ and $P$. We denote the encoding coefficients for the $i$-th  node as:
\begin{align*}
    \bm{\alpha}^{(i)} = [\alpha_1^{(i)}, \cdots, \alpha_m^{(i)} 
    ]^T,~ \bm{\beta}^{(i)} = [\beta_1^{(i)}, \cdots, \beta_m^{(i)} ]^T. 
\end{align*}
Let $\calP_k([P]) = \{ \calS: \calS \subseteq [P], |\calS|=k \}$ and let $\calS_p$ be the $p$-th set in $\calP_k([P])$. In other words, $\calP_k([P])$ is a set of all failure scenarios with $k$ successful nodes out of $P$ nodes. Then, we define $\mbd^{(p)}$ as the vector of decoding coefficients when $\mathcal{S}_{p}$ is the set of successful workers. We define our optimization problem as follows:
\begin{mdframed}
\textbf{Optimization for Approximate Coded Computing}:
\begin{align}
&\min_{\substack{\mbalpha^{(i)}, \; \mbbeta^{(i)}, \; \mbd^{(p)} \\ i= 1,\ldots,n, \\ p =1,\ldots, \binom{P}{k}} } \sum_{p=1}^{\binom{P}{k}} || \mbI_{m \times m} - \sum_{i \in \mathcal{S}_p} d_i^{(p)} \mbalpha^{(i)} \mbbeta^{(i) T} ||_F^2 . \label{eq:alt_opt}
\end{align}
\end{mdframed}

Notice that \eqref{eq:alt_opt} is a non-convex problem, but it is convex with respect to each coordinate, i.e., with respect to $\{\mbalpha^{(i)}:i \in [n]\}$, $\{\mbbeta^{(i)}:i\in [n]\}$, and $\{\mbd^{(p)}:p \in \left[\binom{P}{k}\right]\}$. Hence, we propose an alternating minimization algorithm that minimizes for $\mbd^{(p)}$, $\mbalpha^{(i)}$, and $\mbbeta^{(i)}$ sequentially. Each minimization step is a quadratic optimization with a closed-form solution, which we describe in the following proposition. The notation used in the proposition and in Algorithm~\ref{alg:alt_opt} is summarized in Table~\ref{tab:notations}.

\begin{table}[t]
    \centering
    \begin{tabular}{ccc}
        \toprule
        Symbol & Dimension & Expression \\ 
        \midrule
        $\mbcalA$ & $m\times P$ & $\begin{bmatrix} \mbalpha^{(1)}& \cdots & \mbalpha^{(P)} \end{bmatrix}$ \\ 
        \hline
        $\mbcalB$ &  $m\times P$ & $\begin{bmatrix} \mbbeta^{(1)}& \cdots & \mbbeta^{(P)} \end{bmatrix}$\\ 
        \hline
        $\mbZ^{(\text{full})}$ & $P \times P$ &   $(\mbcalA^T \mbcalA) \odot (\mbcalB^T \mbcalB)$ \\ 
        \hline 
        $\mbz^{(\text{full})}$ & $P$ &  $[\mbalpha^{(i)} \cdot \mbbeta^{(i)}]_{i=1,\ldots,P}$ \\ 
        \hline 
        $\mbZ^{(p)}$ & $k \times k$ & $\mbZ^{(\text{full})}|_{i\in \calS_p, j \in \calS_p}$ \\ \hline
        $\mbz^{(p)}$ & $k \time 1$ & $\mbz^{(\text{full})}|_{i\in \calS_p}$ \\ \hline 
        $\mbY$ & $P \times P$ & $\begin{bmatrix}
         \sum_{p: i,j \in \calS_p} d_i^{(p)} d_j^{(p)} \end{bmatrix}_{\substack{i=1,\ldots,P \\j=1,\ldots, P}}$ \\ \hline
        $\mby$ & $P$ & $\begin{bmatrix}
         \sum_{p: i\in \calS_p} d_i^{(p)}
        \end{bmatrix}_{i=1,\ldots,P}$ \\ \hline
        $\mbY_{\mbcalA}$, $\mbY_{\mbcalB}$ &$P \times P$ & $\mbY_{\mbcalA} = \mbY \odot (\mbcalA^T \mbcalA)$,  $\mbY_{\mbcalB} = \mbY \odot (\mbcalB^T \mbcalB)$ \\ \bottomrule
    \end{tabular}
    \vspace{1em}
    \caption{Summary of Notations used in Proposition~\ref{prop:opt_stable} and Algorithm~\ref{alg:alt_opt}}
    \label{tab:notations}
\end{table}

\begin{proposition}\label{prop:opt_stable}
The stationary points of the objective function given in \eqref{eq:alt_opt} satisfy 
\vspace{-0.8em}
\begin{align*}
    &\text{(i) } \;\mbZ^{(p)} \cdot \mbd^{(p)} = \mbz^{(p)} \;\; \text{ for } \; p= 1, \ldots, \binom{n}{k}, \\ 
    &\text{(ii) } \; \mbY_{\mbcalB} \mbcalA = \mathrm{diag}(\mby) \mbcalB, \\
    &\text{(iii) } \;  \mbY_{\mbcalA} \mbcalB = \mathrm{diag}(\mby) \mbcalA,
\end{align*}
where $\mathrm{diag}(\mby)$ is an $n$-by-$n$ matrix which has $y_i$ on the $i$-th diagonal and $0$ elsewhere. 
\end{proposition}
Proof is given in Appendix~\ref{app:prop_proof}.
Algorithm~\ref{alg:alt_opt} presents an alternating minimization procedure for computing a local minimum of \eqref{eq:alt_opt}. The algorithm sequentially solves conditions (i)--(iii) in Proposition \ref{prop:opt_stable}. Since each step corresponds to minimizing  \eqref{eq:alt_opt} for one of the variables $\mbd^{(p)}$, $\mbcalA$, and $\mbcalB$, the resulting objective is non-increasing in the algorithm's iterations and converges to a local minimum.

\begin{algorithm}[!tb] 
\SetAlgoLined
\textbf{Input:} Positive Integers $m, k$ and $P$ ($P>k$)\;
\textbf{Output:} $\mbcalA$, $\mbcalB$, $\mbd^{(p)}$ ($p=1,\ldots,P$)\;
\textbf{Initialize:} Random $m \times P$ matrices $\mbcalA$ and $\mbcalB$\;

 \While{num\_iter $<$ max\_iter}{
 Compute $\mbZ^{\text{(full)}}$ and $\mbz^{\text{(full)}}$ from $\mbcalA$ and $\mbcalB$\;
 \For{$p\leftarrow 1$ \KwTo $P$}{
     Solve for $\mbd^{(p)}$ :  $\mbZ^{(p)} \mbd^{(p)} = \mbz^{(p)}$
    }
Compute $\mbY$ and $\mby$, and $\mbY_{\mbcalB}$ \;
$\mbcalA \leftarrow \mbcalA^{*}$, $\mbcalA^{*}$: solution of $\mbY_{\calB} \cdot \mbcalA = \mathrm{diag}(\mby) \cdot \mbcalB$\;
Compute $\mbY_{\mbcalA}$ \;
$\mbcalB \leftarrow \mbcalB^{*}$, $\mbcalB^{*}$: solution of $\mbY_{\mbcalA} \cdot \mbcalB = \mathrm{diag}(\mby) \cdot \mbcalA$\;
}
 \caption{Alternating Quadratic Minimization} 
 \label{alg:alt_opt}
\end{algorithm}

We next show how the optimization objective in~\eqref{eq:alt_opt} is related to the relative error of the computation output. Let $\ell^{(p)}$ be the loss function for the $p$-th scenario, i.e., 
\begin{equation}\label{eq:ell_p}
    \ell^{(p)} = || \mbI_{m \times m} - \sum_{i \in \mathcal{S}_p} d_i^{(p)} \mbalpha^{(i)} \mbbeta^{(i) T} ||_F^2.
\end{equation}
\begin{theorem} \label{thm:opt_bound}
The error between the decoded result from the nodes in $\calS_p$, $\widehat{\mbC}_{\calS_p}$,  and the true result $\mbC$ can be bounded as:
\begin{equation}
    ||\mbC -\widehat{\mbC}_{\calS_p}||_F \leq \sqrt{\ell^{(p)}}\cdot m \cdot \eta^2.
\end{equation}
\end{theorem}
The proof is given in Appendix~\ref{app:optboundproof}

\subsection{Optimization Results} \label{subsec:opt_result}

\begin{figure}[t]
\vspace{-2em}
\centering
\subfloat[Min Loss (over 1000 seeds) vs $k$ for $P=k+1$]{\label{fig:res4}\includegraphics[scale=.35] {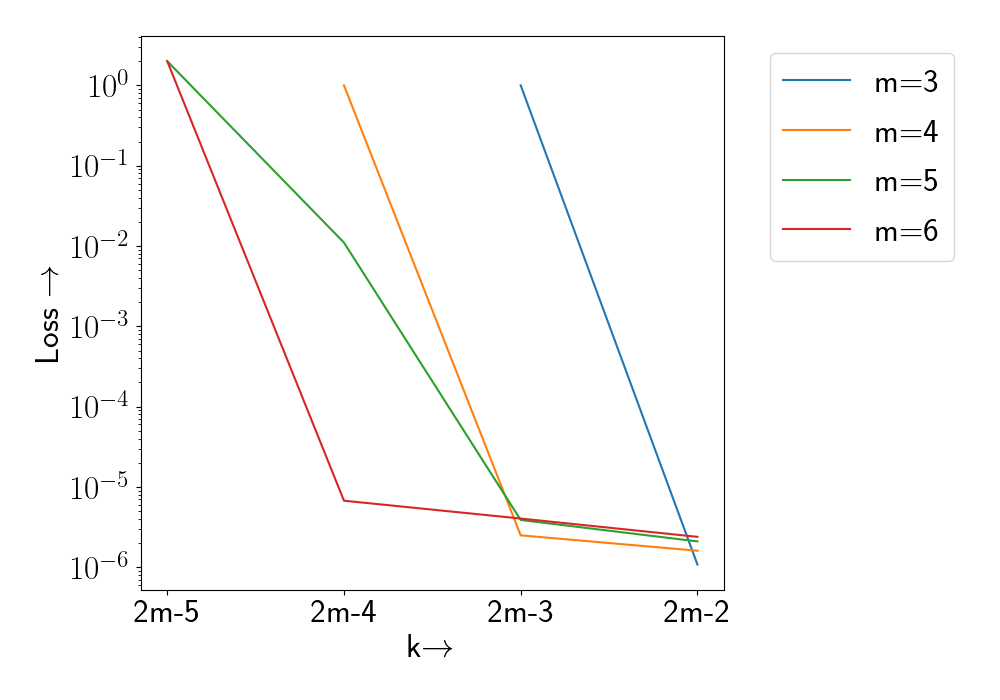}}
\subfloat[Min Loss (over 1000 seeds) vs $m$ for $k=2m-2$]{\label{fig:res1}\includegraphics[scale=.35]{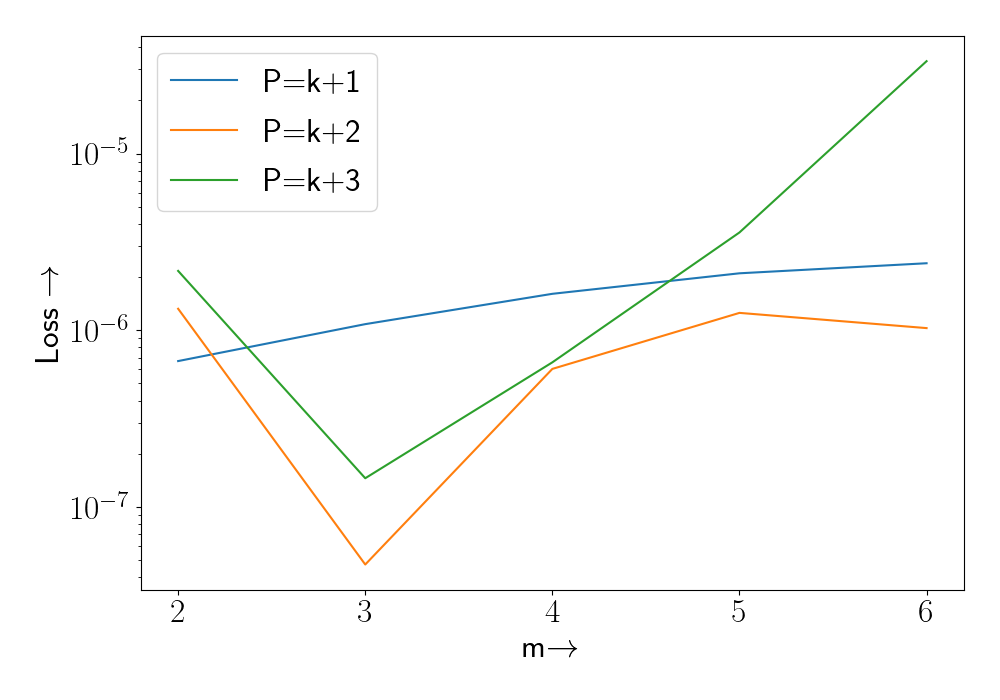}}\qquad
\subfloat[Avg Loss (over 1000 seeds) vs $m$ for $k=2m-2$ with $95\%$ confidence interval]{\label{fig:res7}\includegraphics[scale=.35]{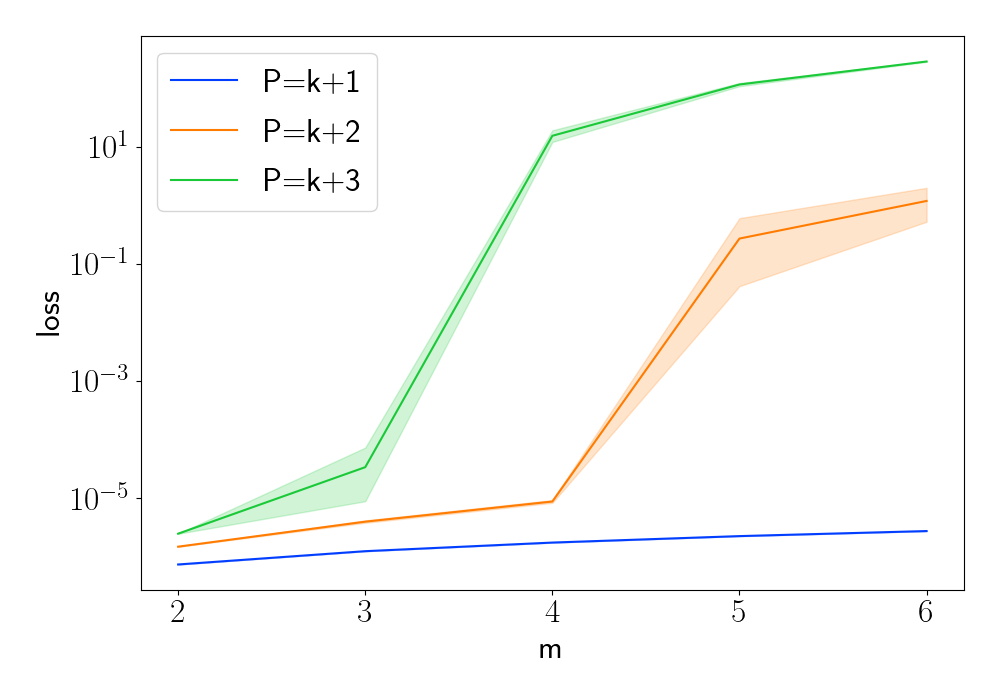}}

  \caption{Summary of results of running Algorithm~\ref{alg:alt_opt} for 1,000,000 iterations. The y-axis is the loss function given in~\eqref{eq:alt_opt}.}
  \label{fig:resopt}
\end{figure}

We summarize the results of running Algorithm~\ref{alg:alt_opt} for various combinations of parameters, $m, k, P$ in Fig.~\ref{fig:resopt}. We report the best result out of 1000 random initializations (seeds); for each trial, we ran Algorithm \ref{alg:alt_opt} for 1,000,000 iterations. 
{ In Fig.~\ref{fig:res4} and Fig.~\ref{fig:res1}, we plot min loss which is the loss from the best code picked from all initializations, and in Fig.~\ref{fig:res7}, we plot the loss averaged over 1000 different random initializations. }
{In Fig.~\ref{fig:res4} where we vary $k$ for fixed $P=k+1$, we observe that the loss increases  with the increase in fault-tolerance, that is, as the recovery threshold $k$ reduces and becomes closer to $m$. Figures \ref{fig:res7} and \ref{fig:res1} plot the loss of the code generated by the optimizer by fixing $k=2m-2$.
Fig. \ref{fig:res1}  demonstrates the best codes found have a loss ($\sim 10^{-5}$) that is much smaller than $1/m^2$, which implies accurate reconstruction due to Theorem \ref{thm:opt_bound}. The figure thus demonstrates the power of the optimization framework.\footnote{The plot in Fig. \ref{fig:res1} is not monotonic. We believe that this could be because of the randomness in the seeds, as only a very few seeds have  small loss (See Fig. \ref{fig:res7} for the average loss which is much higher). Also, we do not know if the loss is monotonic in $m$.}}

\begin{figure}[t]
\centering
  \subfloat[Loss vs $N_\text{succ}$ for $m=3,P=6$]{\label{fig:jres1}\includegraphics[scale=.33] {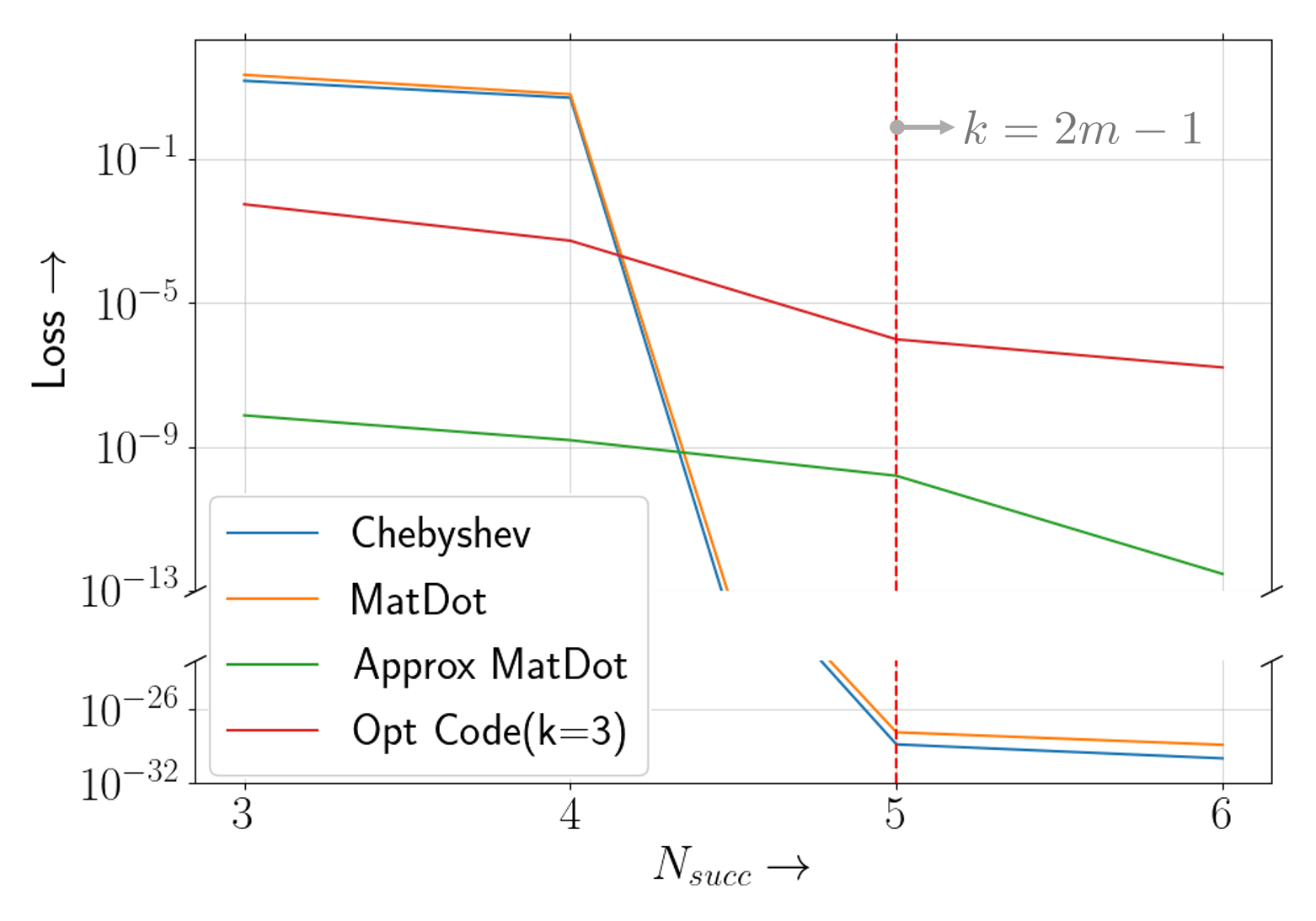}}
  \subfloat[$\epsilon$ vs $N_\text{succ}$ for $m=3,P=6$]{\label{fig:jres2}\includegraphics[scale=.33]{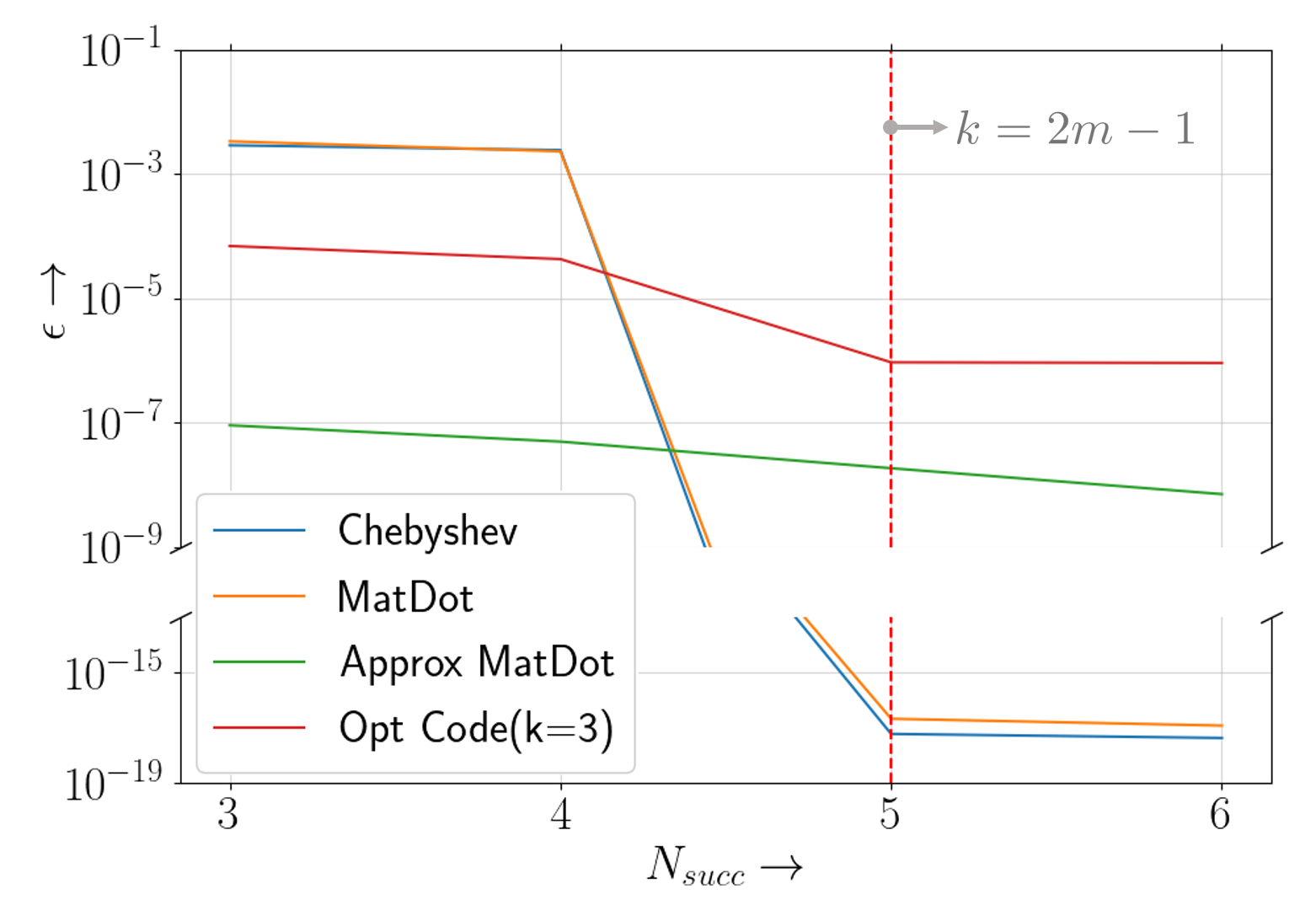}}\qquad
  \subfloat[MatDot: $\epsilon$ vs $\gamma$ for $m=3,P=6$]{\label{fig:jres3}\includegraphics[scale=.33] {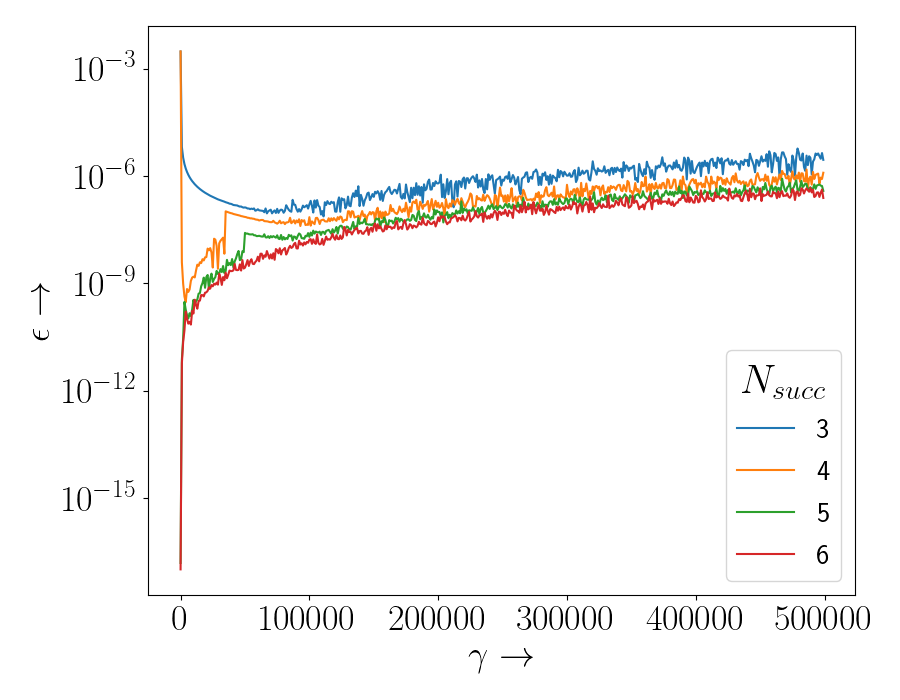}}
  \subfloat[MatDot: $\epsilon$ vs $N_\text{succ}$ for $m=3,P=6$]{\label{fig:jres4}\includegraphics[scale=.33]{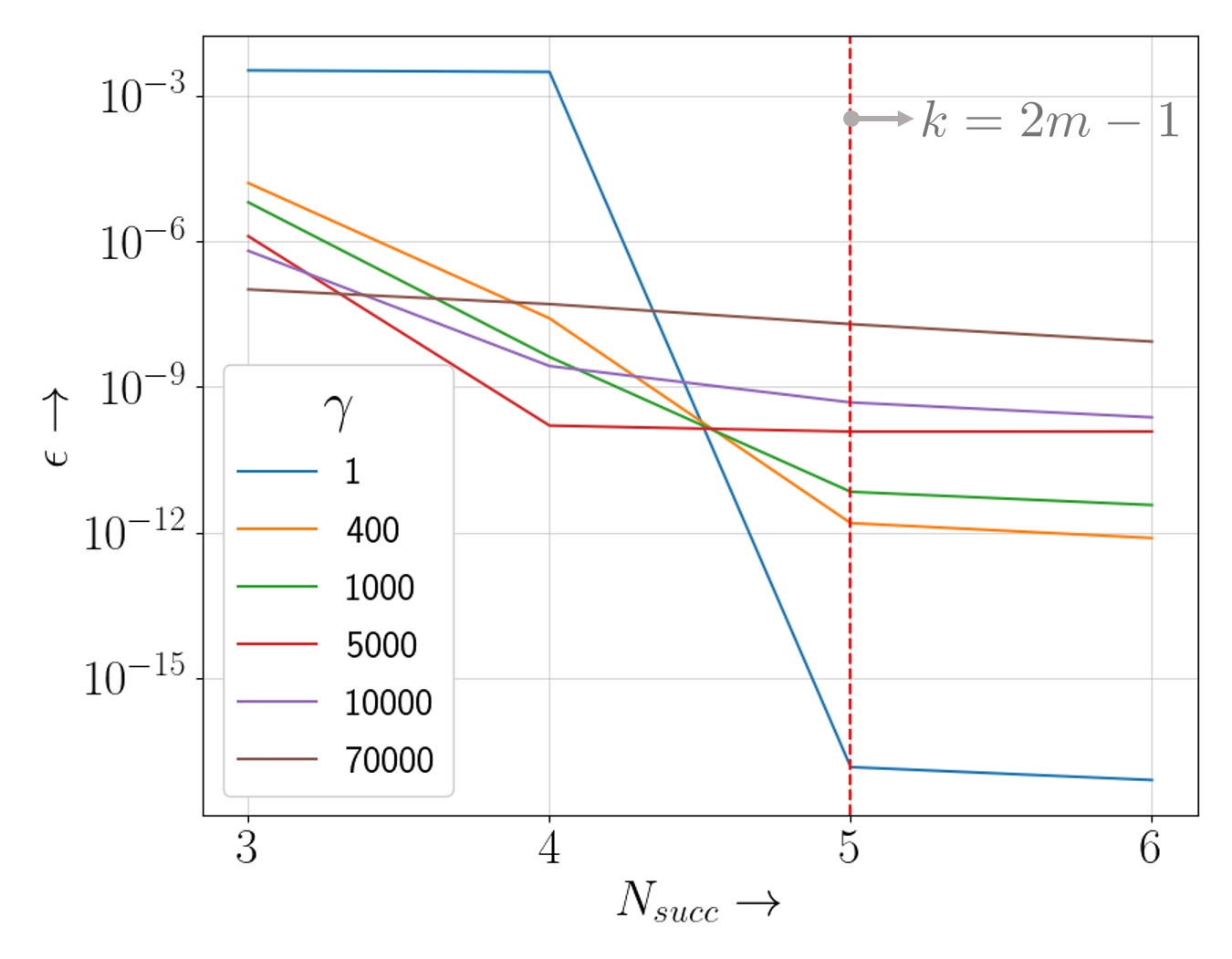}}
  \caption{Performance comparison between different coding methods, over various number of successful nodes for fixed $m =3$ and $P=6$. The y-axis in (a) represents the loss function given in \eqref{eq:alt_opt} and the y-axis in (b),(c),(d) represents the empirical evaluation of $\epsilon$, i.e., $||\widehat{\mbC} - \mbC||_{\max}$. }
  \label{fig:resopt2}
\end{figure}

In Fig. \ref{fig:resopt2}, we compare the performance of the conventional MatDot and Chebyshev polynomial-based codes \cite{fahim2019numerically} with the approximate MatDot codes and optimization codes developed in this paper. {In the figures, the parameter $N_\text{succ}$ represents the number of non-straggling nodes. Note the difference between $N_{succ}$ and $k$; $k$ represents the recovery threshold that the codes have been designed for, while $N_{succ}$ represents the number of nodes that do not straggle and is independent of code design.}
In Fig.~\ref{fig:jres1} the loss is computed in accordance with \eqref{eq:alt_opt}, where $\mbcalA,\mbcalB,\mbY$ are derived from the encoding and decoding procedures of the respective codes. {The $2m-1$ recovery threshold is highlighted in a red dotted vertical line for reference.} In Fig. \ref{fig:jres2}, we show the actual error in the decoded matrix product, i.e., $\epsilon = ||\mbC - \widehat{\mbC}||_{\max}$. To compute this, we performed multiplications of two random unit matrices of sizes $100\times 100$. MatDot and Approximate MatDot codes are constructed using the evaluation points $\bm{\lambda}(1)$ and $\bm{\lambda}(70000)$ respectively, where
\[\bm{\lambda}(\gamma)=\left\{\frac{1}{\gamma}\cos\left(\frac{(2i-1)\pi}{2P}\right)\right\}_{i=1}^P.\]
The above equation is consistent with picking of Chebyshev nodes as our evaluation points. The Chebyshev nodes are a popular choice \cite{trefethen2019approximation}  to mitigate the well-known Runge phenomenon, where the interpolation error increases closer to the boundaries of the interval $[-1/\gamma,1/\gamma].$ It is also instructive to note that the only difference between MatDot and Approx Matdot is the choice of $\gamma$, the encoding scheme is same.

{Both Figures \ref{fig:jres1} and \ref{fig:jres2} demonstrate that for $N_{\text{succ}} \geq k=2m-1$, the Matdot and Chebyshev codes have very small loss and error $(~10^{-16}),$ however, for $N_{\text{succ}} < k$ Approximate MatDot codes and optimized codes outperform Chebyshev and MatDot codes.}

{Figures~\ref{fig:jres3} and \ref{fig:jres4} represent MatDot codes (with evaluation points $\bm{\lambda(\gamma)}$) behavior for increasing condition number (controlled with $\gamma$ parameter). Observe that the error 
$\epsilon$ is composed of two quantities: $\epsilon_1,$ the  interpolation error under infinite precision,  and $\epsilon_2$ the computation error due to finite precision. For $N_{\text{succ}} \geq k=2m-1$, $\epsilon_1=0$ for MatDot, Approximate-MatDot and Chebyshev Codes. $\epsilon_2$ increases as $\gamma,$ and therefore, MatDot and Chebyshev codes have lower loss. However, for $N_{\text{succ}} \leq K=2m-1$, $\epsilon_1$ decreases and $\epsilon_2$ increases, as $\gamma$ increases; therefore, the error is non-monotonic in $\gamma$. These phenomena are transparent from Figures \ref{fig:jres3}, \ref{fig:jres4}.}
 The source code for Figure~\ref{fig:resopt} and \ref{fig:resopt2} are in \cite{GitHubRepo_approx}.

\section{Application} \label{secc:app}

In this section, we illustrate that approximate coded computing is particularly useful for training  machine learning (ML) models.  ML models are usually trained using optimization algorithms that have inherent stochasticity (e.g., stochastic gradient descent). These algorithms are applied to finite, noisy training data. Consequently, ML models can be  tolerant of the accuracy loss resulting from approximate computations during training. In fact, this loss can be insignificant when compared to other factors that impact training performance (parameter initialization, learning rate, dataset size, etc.). 

We illustrate this point by considering a simple logistic regression training scenario modified to use coded computation. 
First, we describe how coded matrix multiplication strategies can be applied to training a logistic regression model. Then, we train a model on the MNIST dataset \cite{lecun-mnisthandwrittendigit-2010} using approximate coded computing strategies and show that the accuracy loss due to approximate coded matrix multiplication is very small.

\subsection{Logistic regression model with coded computation}
We consider logistic regression with cross entropy loss and softmax function. We identify parts of training steps where coded computation could be applied.

Consider a dataset 
$ \mathcal{D} = \{(\bm{x}_1,\bm{y}_1) \ldots (\bm{x}_D, \bm{y}_D) \} $
and the loss function $L(\bm{W}; \mathcal{D})$ with gradient $\frac{\partial}{\partial \bm{W}} L(\bm{W}; \mathcal{D})$, for model $\bm{W}$. Let there be $J$ classes in the dataset, and $\{\bm{y}_i\}_{i=1}^D$ be a set of one-hot encoded vectors, such that $y_{ji}=1$ means $i$\textsuperscript{th} data point $\bm{x}_i$ belongs to $j$\textsuperscript{th} class. Let $\bm{Y}\in\mathbb{F}_2^{J\times D}=[\bm{y}_1, \dots, \bm{y}_D]$. Let $\bm{W}=[\bm{w}_1;\dots;\bm{w}_J]$ ($\bm{w}_j$ is a row vector) be a matrix that comprises the logistic regression training parameters. The cross entropy loss is given by:
\begin{equation}
    L(\bm{W}; \mathcal{D}) = \sum_{i=1}^{D}\sum_{j=1}^{J}y_{ji}\log p(y_{ji}=1|\bm{x}_i)
    \label{eq:logisticloss}
\end{equation}
where
\begin{equation*}
    p(y_{ji}=1|\bm{x}_i) = \text{softmax}(z_{ji})=\frac{e^{z_{ji}}}{\sum_{j=1}^{J}e^{z_{ji}}}, \;\; z_{ji}=\bm{w}_j\bm{x}_i
\end{equation*}
$\bm{x}_i$ is an column vector. Define $\bm{Z}\in\mathbb{R}^{J\times D}=\{z_{ji}\}_{j=1,i=1}^{J,D}$ and $\bm{X}=\{\bm{x}_i\}_{i=1}^D$. Then we write above equation as \begin{equation}
    \bm{Z}=\bm{W}\bm{X}
    \label{eq:lgapply1}
\end{equation}
The gradient is computed as:
\begin{align}
    \frac{\partial}{\partial \bm{W}} L(\bm{W}; \mathcal{D} ) &=\bm{H}\bm{X}^T
    \label{eq:lgapply2}
\end{align}

where $\bm{H} = (\text{softmax}(\bm{Z})-\bm{Y}),$ and we apply softmax function element-wise.

Clearly, we can apply the coded matrix multiplication schemes to computations in \eqref{eq:lgapply1} and \eqref{eq:lgapply2}. In \eqref{eq:lgapply1}, we encode $\bm{W}$ and $\bm{X}$ and perform coded matrix multiplication, then we encode $\bm{H}$ and reuse encoded $\bm{X}$ to perform another coded matrix multiplication in \eqref{eq:lgapply2}.

\subsection{Results}
The goal is to explore whether, despite the loss of precision due to approximation, our approach leads to accurate training. We trained the logistic regression using the MNIST dataset \cite{lecun-mnisthandwrittendigit-2010}. A learning rate of 0.001 and batch size of 128 were used. Each logistic regression experiment was run for 40,000 iterations. {For every matrix multiplication step in the training algorithm i.e., computing $\bm{WX}$ and $\bm{HX}^T$, we assume that we have $k$ successful nodes out of $P$ nodes.} 
{Tables \ref{table:reslg1} and \ref{table:reslg2} show the 10-fold cross validation accuracy results obtained for training and test datasets. For accurate comparison, we used the same the random folds, initialization and batches for the different coding schemes.} 

We first ran the training algorithm for the worst-case failure scenario where we assume that the worst-case failure pattern happens at every multiplication step, i.e., out of $\binom{P}{k}$ failure scenarios, we always have $p_\text{worst} = \argmax_{p \in \{1, \ldots, \binom{P}{k} \}} \ell^{(p)}$. The results are summarized in Table~\ref{table:reslg1}. {To simulate a scenario where different nodes straggle in different iterations, we ran the experiments for a scenario where a random subset of $k$ nodes returns at every iteration. The corresponding accuracies are given in Table~\ref{table:reslg2}.} 
For each coding scheme, we fit the corresponding encoding and decoding matrices into $\mbalpha$, $\mbbeta$, $\mbd^{(p)}$. For example, for approximate MatDot, Vandermonde matrices are used in $\mbalpha$ and $\mbbeta$, and its corresponding decoding coefficients are put in $\mbd^{(p)}$ and "Opt Code" represents the codes obtained from the optimization algorithm \ref{alg:alt_opt}. $k$ represents the recovery threshold.
The training and test accuracies for uncoded strategy (without failed nodes) are 92.32±0.07 and 91.17±0.25. 
The training and test accuracies for uncoded distributed strategy (with 2 failed nodes, but no error correction) are 21.45±1.00 and 21.48±1.05 for $m=5$, and 29.09±5.99 and 28.96±6.32 for $m=20$. {This indicates the importance of redundancy and error correction for accurate model training in presence of stragglers.}
%
\begin{table}[ht]
\small
    \centering
    \begin{tabular}{c c c c c c c }
    \toprule
    & \multicolumn{3}{c}{Training Accuracy (\%)} & \multicolumn{3}{c}{Test Accuracy (\%)} \\    
    \cmidrule(lr){2-4}\cmidrule(lr){5-7}
      $(m,k,P)$ & Approx MatDot & Chebyshev & Opt Code & Approx MatDot & Chebyshev & Opt Code\\ \midrule
      (5,5,7) & 92.32±0.07 & 29.00±2.84 & 16.95±1.63 & 91.17±0.26 & 29.14±2.99 & 16.96±1.55\\ \hline
(5,6,8) & 92.33±0.07 & 41.83±4.90 & 92.16±0.10 & 91.16±0.25 & 41.81±5.38 & 91.05±0.30\\ \hline
(5,7,9) & 92.33±0.07 & 50.55±8.00 & 92.32±0.07 & 91.17±0.25 & 50.44±8.22 & 91.16±0.25\\ \hline
( 5, 8,10) & 92.32±0.07 & 47.10±5.67 & 92.32±0.08 & 91.17±0.25 & 46.82±5.63 & 91.16±0.25\\ \hline
( 5, 9,11) & 92.32±0.07 & 92.32±0.07 & 92.32±0.07 & 91.17±0.25 & 91.17±0.25 & 91.17±0.25\\ \hline
(20,20,22) & 44.25±2.06 & 38.50±8.26 & 73.16±3.17 & 44.04±1.38 & 38.33±8.16 & 72.69±3.26\\ \bottomrule
    \end{tabular}
    \caption{Logistic regression results on MNIST dataset for worst case failures}
    \label{table:reslg1}
\end{table}
\begin{table}[ht]
\small
    \centering
    \begin{tabular}{c c c c c c c }
    \toprule
    & \multicolumn{3}{c}{Training Accuracy (\%)} & \multicolumn{3}{c}{Test Accuracy (\%)} \\    
    \cmidrule(lr){2-4}\cmidrule(lr){5-7}
      $(m,k,P)$ & Approx MatDot & Chebyshev & Opt Code & Approx MatDot & Chebyshev & Opt Code\\ \midrule
    (5,5,7) & 92.32±0.07 & 91.37±0.10 & 88.36±0.39 & 91.17±0.25 & 91.03±0.34 & 87.98±0.61\\ \hline
(5,6,8) & 92.32±0.07 & 91.91±0.07 & 92.32±0.08 & 91.17±0.25 & 91.47±0.25 & 91.16±0.27\\ \hline
(5,7,9) & 92.32±0.07 & 92.14±0.10 & 92.32±0.07 & 91.17±0.25 & 91.58±0.16 & 91.17±0.25\\ \hline
( 5, 8,10) & 92.32±0.07 & 92.22±0.09 & 92.32±0.07 & 91.17±0.25 & 91.45±0.23 & 91.17±0.25\\ \hline
( 5, 9,11) & 92.32±0.07 & 92.32±0.07 & 92.32±0.07 & 91.17±0.25 & 91.17±0.25 & 91.17±0.25\\ \hline
(20,20,22) & 46.05±1.47 & 92.54±0.05 & 92.44±0.10 & 45.97±1.02 & 91.85±0.19 & 91.66±0.18\\ \bottomrule
    \end{tabular}
    \caption{Logistic regression results on MNIST dataset for random failures}
    \label{table:reslg2}
\end{table}

{From the results in both tables, we observe that MatDot codes have essentially identical performance to uncoded multiplication for smaller values of parameter $m$. The Opt codes also give approximately identical results as uncoded computation, and for bigger values of $m$ Opt codes appear to have better performance than MatDot codes.\footnote{The only exception to this statement is the $(5,5,7)$ case where Opt codes perform poorly. Here, it is possible that a more expansive search across a larger set of random seeds  leads to a code with comparable performance as uncoded computation.} As expected, the random failure scenario in Table \ref{table:reslg2} has much better performance than the worst-case failure scenario.} The codes for logistic regression implementation via uncoded and coded multiplications can be found in \cite{GitHubRepo_approx}.

\section{Discussion and Future Work}

This paper opens new directions for coded computing by showing the power of approximations. Specifically, an open research direction is the investigation of related coded computing frameworks (e.g., polynomial evaluations) to examine the gap between $\epsilon$-error and $0$-error recovery thresholds. {The optimization approach provides a practical framework to find good codes for general approximate coded computing problems beyond matrix multiplication. Further, the framework can guide the development of coded computing theory by giving heuristic insights into performance; in fact, the main theoretical results of our paper (Theorems \ref{thm:approx_matdot}, \ref{thm:approx_polydot}) were the result of such hints that were provided by the framework. From a practical viewpoint, it is an open question as to whether matrix multiplication codes developed by the optimization framework are better than the codes developed via theory, i.e., the approximate MatDot codes. While the results of Sec. \ref{secc:app} indicate that approximate MatDot codes are better for smaller values of the parameter $m$ and the optimization framework seems to have better performance for larger values of $m$, it is unclear whether this behavior is fundamental.}

As our constructions require evaluation points close to $0$ (Section \ref{sec:approx_matdot}), encoding matrices  become ill-conditioned rapidly as $m$ grows. An open direction of future work is to explore numerically stable coding schemes possibly building on recent works, e.g. \cite{fahim2019numerically, ramamoorthy2019numerically,subramaniam2019random}, with  focus on $\epsilon$-error instead of exact computation. 

\begin{singlespace}
\bibliographystyle{IEEEtran}
\bibliography{IEEEabrv,references,cadambe-references}
\end{singlespace}

\appendices 
\section{Proof of Proposition~\ref{prop:opt_stable}} \label{app:prop_proof}
For simple demonstration, let us focus on the case where $\calS_p = [1, \ldots, k]$ and expand the term inside the sum. In the following equations, we will omit the superscript $(p)$ for simplification. 

\begin{align*}
    || \mbI_{m \times m} - \sum_{i \in \calS_p} d_i \mbalpha^{(i)} \mbbeta^{(i) T} ||_F^2 &=   \Tr \left( \left( \mbI - \sum_{i=1}^{k} d_i \mbX_i\right)^T \left( \mbI - \sum_{i=1}^{k} d_i \mbX_i \right)  \right) \quad (\mbX_i = \mbalpha^{(i)} \mbbeta^{(i)T}) \\ 
    &=  \Tr \Bigg(\mbI - \sum_{i=1}^{k} d_i \mbX_i^T -\sum_{i=1}^{k} d_i \mbX_i + \left(\sum_{i=1}^{k} d_i \mbX_i \right)^T  \left(\sum_{j=1}^{k} d_j \mbX_j \right) \Bigg) \\ 
    &= m - 2\sum_{i=1}^{k} d_i \Tr(\mbX_i) + \sum_{i=1}^{k} \sum_{j=1}^{k} d_i d_j \Tr (\mbX_i^T \mbX_j)  \\
    &= m - 2\sum_{i=1}^{k} d_i \mbalpha^{(i)} \cdot \mbbeta^{(i)}  + \sum_{i=1}^{k} \sum_{j=1}^{k} d_i d_j (\mbalpha^{(i)}\cdot \mbalpha^{(j)})(\mbbeta^{(i)}\cdot \mbbeta^{(j)}) \\
      &= m - 2 \mbd \cdot \mbz +  \mbd^T \mbZ \mbd,
\end{align*}
where $\mbd$ and $\mbz$ are length-$k$ column vectors: $\mbd= [d_1, \ldots, d_k]$ and 
$\mbz = [\mbalpha^{(i)} \cdot \mbbeta^{(i)}]_{i=1,\ldots,k}$. $\mbZ$ is a $k \times k$ matrix: 
$
    \mbZ = (\mbcalA_k^T \mbcalA_k) \odot (\mbcalB_k^T \mbcalB_k), 
$
where $\mbcalA_k = \begin{bmatrix} \mbalpha_{1} & \mbalpha_{2} & \ldots & \mbalpha_k \end{bmatrix}$ and $\mbcalB_k = \begin{bmatrix} \mbbeta_{1} & \mbbeta_{2} & \ldots & \mbbeta_k \end{bmatrix}$.

The partial derivative with respect to $\mbd$ can be represented as: 
$
    \frac{\partial}{\partial \mbd} L = -2 \mbz + 2\mbZ \mbd.
$
Thus, the optimal $\mbd^{*}$ can be obtained by solving:
$
    \mbZ \mbd = \mbz. 
$
Note that this can be easily generalized to any $\mbd^{(p)} = \begin{bmatrix} d_i \end{bmatrix}_{i \in \calS_p}$. It only requires using different $\mbZ$ and $\mbz$ as follows: 
\begin{equation*}
    \mbZ = (\mbcalA^{(p) T} \mbcalA^{(p)}) \odot (\mbcalB^{(p) T} \mbcalB^{(p)}), \;\;  \mbz^{(p)} = [\mbalpha^{(i)} \cdot \mbbeta^{(i)}]_{i \in S_p},
\end{equation*}
where $\mbcalA^{(p)} = [\mbalpha_i]_{i \in S_p}, \; \mbcalB^{(p)} = [\mbbeta_i]_{i \in S_p}$. 

To obtain the gradient with respect to $\mbalpha^{(i)}$'s and $\mbbeta^{(i)}$'s, let us expand the loss function given in \eqref{eq:alt_opt} again. We now want to include the outer sum: 
\begin{align*}
     L &=  \sum_{p=1, \ldots, \binom{n}{k}} || \mbI_{m \times m} - \sum_{i \in \calS_p} d_i^{(p)} \mbalpha^{(i)} \mbbeta^{(i) T} ||_F^2 \\ 
     &= \sum_{p=1, \ldots, \binom{n}{k}} \Big(m - 2\sum_{i \in \calS_p} d_i^{(p)} \mbalpha^{(i)} \cdot \mbbeta^{(i)} + \sum_{i\in \calS_p} \sum_{j \in \calS_p} d_i^{(p)} d_j^{(p)} (\mbalpha^{(i)}\cdot \mbalpha^{(j)})(\mbbeta^{(i)}\cdot \mbbeta^{(j)})\Big) \\
     &= \binom{n}{k} \cdot m -2 \sum_{i \in [n]} \big(\sum_{p: i\in \calS_p} d_i^{(p)} \big) \mbalpha^{(i)}\cdot \mbbeta^{(i)} + \sum_{i \in [n]} \sum_{j \in [n]} \big( \sum_{p: i,j \in \calS_p } d_i^{(p)} d_j^{(p)} \big) (\mbalpha^{(i)}\cdot \mbalpha^{(j)}) (\mbbeta^{(i)}\cdot \mbbeta^{(j)}) \\ 
     &= \binom{n}{k} \cdot m -2 \sum_{i \in [n]} y_i \cdot \mbalpha^{(i)}\cdot \mbbeta^{(i)} + \sum_{i \in [n]} \sum_{j \in [n]} Y_{i,j} (\mbalpha^{(i)}\cdot \mbalpha^{(j)}) (\mbbeta^{(i)}\cdot \mbbeta^{(j)})
\end{align*}
In the last line, we let $y_i = \sum_{p: i\in \calS_p} d_i^{(p)}$ and $Y_{i,j} =  \sum_{p: i,j \in \calS_p } d_i^{(p)} d_j^{(p)}$. 
Now, the gradient of $L$ with respect to $\mbalpha^{(i)}$ can be written as: 
\begin{equation}
    \frac{\partial}{\partial \mbalpha^{(i)}} L = -2y_i \mbbetai + 2\sum_{j \in [n]} Y_{i,j} (\mbbetai \cdot \mbbetaj) \mbalphaj. 
\end{equation}
Following the notation that $\mbcalA = \begin{bmatrix} \mbalpha^{(1)}& \ldots & \mbalpha^{(n)} \end{bmatrix}$ and $\mbcalB =\begin{bmatrix} \mbbeta^{(1)}& \ldots & \mbbeta^{(n)} \end{bmatrix} $, this can be written in a matrix form:
$
    \frac{\partial}{\partial \mbcalA} L = -2 \mathrm{diag} (\mby) \mbcalB^T + 2 \mbY_{\mbcalB} \mbcalA^T, 
$
where $\mby = [y_i]_{i=1,\ldots,n}$ is a column vector of length $n$ and $\mbY_{\calB} = [Y_{i,j} (\mbbetai \cdot \mbbetaj)]_{i,j=1,\ldots,n} = \mbY \odot (\mbcalB^T \mbcalB)$ is an $n \times n$ matrix. Thus, the optimal $\mbcalA^{*}$ can be obtained by solving: 
$
    \mbY_{\calB} \cdot \mbcalA = \mathrm{diag}(\mby) \cdot \mbcalB.
$
Similarly, the optimal $\mbcalB^{*}$ can be obtained by solving:
$
    \mbY_{\calA} \cdot \mbcalB = \mathrm{diag}(\mby) \cdot \mbcalA,
$
where $\mbY_{\calA} = \mbY \odot (\mbcalA^T \mbcalA)$.

\section{Proof of Theorem~\ref{thm:approx_matdot}} \label{app:approx_matdot}
Let $f$ be a $(k-1)$-degree polynomial
$
    f(x) = a_0 + a_1 x + \cdots + a_{k-1} x^{k-1}, 
$
and we use  $\mathsf{vec}(f)$ to denote the row vector representation of the coefficients of $f$, i.e., 
\begin{equation}
    \mathsf{vec}(f) \triangleq  \begin{bmatrix} a_0 & a_1 & \cdots & a_{k-1} \end{bmatrix} \triangleq \mba.
\end{equation}
Also, let $\mblambda \triangleq [\lambda_j]_{j\in [m]} \in \mathbb{R}^m$  be $m$ distinct evaluation points. We define a Vandermonde matrix for $\mblambda$ of degree $k-1$ as:
\begin{equation}
    \mathsf{Vander}(\mblambda, k) = \begin{bmatrix}
    1 & \cdots & 1 \\ 
    \lambda_1 & \cdots & \lambda_m \\
    \vdots & \ddots & \vdots \\ 
    \lambda_1^{k-1} & \cdots & \lambda_m^{k-1}
    \end{bmatrix}_{k \times m}.
\end{equation}
The evaluations of $f$ at the points $\mblambda$ can be written as
\begin{equation} \label{eq:a_V}
    \begin{bmatrix} f(\lambda_1) & f(\lambda_2) & \cdots & f(\lambda_m) \end{bmatrix} = \mba\cdot 
     \mbV,
\end{equation}
where  $\mbV = \mathsf{Vander}(\mblambda, k)$.
When $m<k$, the null space of $\mbV$ can be conveniently expressed in terms of elementary symmetric polynomials.

\begin{definition}[Elementary Symmetric Polynomial]
Let $\mbx = (x_1, \cdots, x_n)$. For $l\in \{0,1,\dots,n\}$, the elementary symmetric polynomials in $n$ variables $e_l:\mathbb{R}^n\to\mathbb{R}$ are given by
\begin{equation}
    e_l(\mbx) \triangleq \begin{cases} \displaystyle \sum_{\substack{S \subseteq [n] \\ |S| = l}} \prod_{i \in S} x_i, & \text{ if } l = 1, \ldots, n, \\ 
    1, & \text{ if } l = 0.
    \end{cases}
\end{equation}
In particular, $e_1(\mbx) = \sum_{i\in [n]}x_i$ and $e_n(\mbx) = \prod_{i\in [n]}x_i$.
\end{definition}

\begin{lemma}
For $m < k$, the left null space of $\mbV$ is spanned by $\left\{\mbu_i\right\}_{i\in [k-m]}\subset \mathbb{R}^k$, where $\mbu_i = \mathsf{vec}(p_i)$ for the polynomials $p_i$'s defined as:
\begin{equation}
    p_i(x) \triangleq x^{i-1} \prod_{j=1}^{m} ( x-\lambda_j).
\end{equation}
\end{lemma}

\begin{proof}
First, note that $\mbu_i \in \mathsf{null}(\mbV^T)$: 
$
    \mbu_i \cdot \mbV = \begin{bmatrix}
     p_i(\lambda_1) & \cdots & p_i(\lambda_m)
    \end{bmatrix} = \mathbf{0}.
$
Next, we show that $\mathsf{dim}(\mathsf{span}(\mbu_1, \cdots, \mbu_{k-m}))=k-m$.
The coefficients of the Lagrange polynomial $p_1(x) = \prod_{j=1}^{m} (x-\lambda_j)$ can be written as: 
\begin{equation}\label{eq:b1}
    \mbu_1 = \begin{bmatrix}
     e_m(\mblambda) & \cdots & e_{0}(\mblambda) & 0 & 0 & \cdots 0
    \end{bmatrix},
\end{equation}
with $k-m-1$ trailing zeros, and 
\begin{align}
    \mbu_2 &= \begin{bmatrix}
     0 & e_m(\mblambda) & \cdots & e_{0}(\mblambda) & 0 & \cdots 0
    \end{bmatrix}, \nonumber \\ 
    &\;\;\vdots \nonumber \\
    \mbu_{k-m} &= \begin{bmatrix}
     0 & 0 & \cdots & 0 & e_m(\mblambda) & \cdots & e_0(\mblambda) 
    \end{bmatrix}.  \label{eq:b_2_km}
\end{align}
Let $\mbU \in \mathbb{R}^{(k-m)\times k}$ be the matrix obtained by concatenating $\mbu_i$ row-wise, i.e., the matrix with $i$-th row equal to $\mbu_i$. Note that $\mbU = \left[\mbU_1~\mbU_2 \right]$, where $\mbU_2\in \mathbb{R}^{(k-m)\times (k-m)}$ is a lower-triangular matrix with diagonal entries equal to $e_0(\mblambda)=1$. Consequently, $\mbU_2$ is full-rank (in particular, $\det (\mbU_2)=1$). Therefore  $\mbU$ is also full-rank and $
\mathsf{rank}(\mbU)=\mathsf{dim}(\mathsf{span}(\mbu_1, \cdots, \mbu_{k-m}))=k-m.
$
\end{proof}
We prove next a bound on the evaluation of the elementary symmetric polynomials $e_{l}, l \neq 0$ in terms of the $\ell_{\infty}$-norm of its entries.
\begin{lemma} \label{lem:e_l_bound}
Let $0 < \epsilon \leq 2$ and  $\mbx\in \mathbb{R}^n$. If $\|\mbx\|_\infty \leq \epsilon/n$, then  $|e_l(\mbx)| \leq \epsilon$ for $l \in \{1, 2, \ldots, n \}$.
\end{lemma}
%
\begin{proof}
\begin{align*}
    |e_l(\mbx)| &= \Big| \sum_{\substack{S \in [n] \\ |S| = l}} \prod_{j \in S} x_j \Big| \leq \sum_{\substack{S \in [n] \\ |S| = l}} \Big| \prod_{j \in S} x_j \Big| \leq \sum_{\substack{S \in [n] \\ |S| = l}} \left(\frac{\epsilon}{n}\right)^l = {n \choose l}  \cdot \left(\frac{\epsilon}{n}\right)^l \\ 
    & \leq \frac{n^l}{l!} \cdot \frac{\epsilon^l}{n^l} = \frac{\epsilon^l}{l!} = \epsilon \cdot \prod_{k=2}^{l} \frac{\epsilon}{k} \leq \epsilon. 
\end{align*}
\end{proof}

If $m=k$, the coefficients of $f$ can be recovered \textit{exactly} from  $\left\{f(\lambda_j)\right\}_{j\in [m]}$   by inverting the linear system \eqref{eq:a_V} as long as the evaluation points are distinct. When $m<k$ then, in general, $\mba$ cannot be recovered exactly. In this case, the system \eqref{eq:a_V} is undetermined: denoting  the true (but unknown) coefficient vector by $\mba^*$, any vector in the set
$
 \left\{\mba\right\} + \mathsf{null}(\mbV^T) = \{ \mba + \mbn \; | \; \mbn \in \mathsf{null}(\mbV^T) \}
$
will be consistent with the $m$ evaluation points. Nevertheless, we show next that if the  coefficients of  $f$ have bounded norm, i.e.,
$
    \mba \in \calB_R \triangleq \{ \mbx \in \mathbb{R}^{k} \text{ s.t. } ||\mbx||_2 \leq R \},
$
then the first $m$ coefficients $a_0,\dots,a_{m-1}$ can be approximated with \textit{arbitrary  precision}  by computing $f$ at $m$ distinct and sufficiently small evaluation points. This result is formally stated in Corollary~\ref{cor:main}, which is the main tool  for proving the approximate coded computing recovery threshold.

\begin{theorem}\label{thm:V_null}
Let $\mblambda$ be $m<k$ distinct evaluation points with corresponding  $k\times m$ Vandermonde matrix $\mbV =\mathsf{Vander}(\mblambda, k-1)$ and $0 < \epsilon \leq \min (2, 3R)$. If $\|\mblambda\|_\infty <\frac{\epsilon}{3 R(k-m)m}$, then for any $\mbx \in \calB_{R}\cap \mathsf{null}(\mbV^T)$,
\begin{equation}
    |x[i]| \leq \epsilon \;\; \text{ for } \; i \in [m]. 
\end{equation}
\end{theorem}

\begin{proof}
Since $\mbx \in \mathsf{null}(\mbV^T)$, we can express $\mbx$ as: 
\begin{equation} \label{eq:a_alpha_b}
    \mbx = \alpha_1 \mbu_1 + \cdots + \alpha_{k-m} \mbu_{k-m},
\end{equation}
for some $\alpha_1, \cdots, \alpha_{k-m} \in \mathbb{R}$. For a shorthand notation, we will use $e_l$ for $e_l(\mblambda)$, and let $e_l = 0$ if $l < 0$ or $l > m$. By substituting \eqref{eq:b1}, \eqref{eq:b_2_km} into \eqref{eq:a_alpha_b}, we get:
$
    x[i] = \sum_{j=1}^{k-m} \alpha_j e_{m-i+j} \;\; \text{ for } \; i = 1, \ldots, k. 
$
Since $e_0 = 1$,  
$
    x[k] = \alpha_{k-m} e_0 = \alpha_{k-m}. 
$
Furthermore, because $\mbx \in \calB_R$, $|x[k]| = |\alpha_{k-m}|  \leq R. $ Similarly, 
$x[k-1]  = \alpha_{k-m-1} e_0 + \alpha_{k-m} e_1 = \alpha_{k-m-1}  + \alpha_{k-m} e_1. 
$

Now note that from Lemma~\ref{lem:e_l_bound}, 
$
    |e_l(\mblambda)| \leq \frac{\epsilon}{3R(k-m)} \triangleq \delta. 
$
for $ l \in [m]$. Thus,
\begin{align*}
    |\alpha_{k-m-1}| &= |x[k-1] + \alpha_{k-m} e_1 |  \leq | x[k-1] | + | \alpha_{k-m} e_1 | \leq R + R \cdot \delta = R (1+\delta).
\end{align*}


By repeating the same argument up to $x[m+1]$, for $l =1, \ldots, k-m$, we obtain:
\begin{align}
|\alpha_{l}|  \leq R (1 + \delta)^{k-m-l} 
 \leq R ( 1+ \frac{1}{k-m})^{k-m-l}  
 \leq R ( 1+ \frac{1}{k-m})^{k-m} 
 \leq 3R. 
\end{align}
The second inequality follows from the assumption that  $\epsilon \leq 3R$ and the last inequality holds because for a positive integer $n$: 
$
    (1+\frac{1}{n})^n \leq 3- \frac{1}{n} < 3.
$

Now, for $i \in [m]$, $x[i]$ can be written as:
\begin{align}
    |x[i]| & = |\sum_{j=1}^{k-m} \alpha_j e_{m-i+j}|
    \leq \sum_{j=1}^{k-m} |\alpha_j| |e_{m-i+j}| 
    \leq \sum_{j=1}^{k-m} |\alpha_j| \delta 
    \leq (k-m) \cdot 3R \cdot \delta =  \epsilon. \label{eq:claim1_prf1}
\end{align}
The inequality in \eqref{eq:claim1_prf1} holds because $m-l+j \neq 0$ for $l \in [m]$, and thus $|e_{m-l+j}| \leq \delta$. 
\end{proof}

\begin{corollary} \label{cor:main}
Consider a set $\left\{f(\lambda_j) \right\}_{j\in[m]}$ of $m<k$ evaluations of $f$ at distinct points $\mblambda$. If  $0<\epsilon< \min(2, 3R)$ and $\|\mblambda\|_\infty <\frac{\epsilon}{6R(k-m)m}$, then for any two coefficient vectors $\mba,\mbb\in \calB_R$ that satisfy \eqref{eq:a_V} (i.e., that are consistent with the evaluations), we have 
\begin{equation}
    |a[i] - b[i]| \leq \epsilon ~\text{for } i \in [m].
\end{equation}
\end{corollary}
\begin{proof}
Under the assumptions of the corollary, $\mbn\triangleq \mba-\mbb \in \mathsf{null}(\mbV^T).$ Moreover, the triangle inequality yields: $\|\mbn\|\leq \|\mba\|+\|\mbb\|\leq 2R.$ I.e., $\mbn \in \calB_{2R} \cap \mathsf{null}(\mbV^T)$. The result  follows by a direct application of Theorem~\ref{thm:V_null}.
\end{proof}



For an $n$-by-$n$ matrix $\mbC$, the polynomial $p_{\mbC}(x)$ is essentially a set of $n^2$ polynomials, having one polynomial for each $C[i,j]$ ($i,j \in [n]$). For decoding, we have to interpolate each of those $n^2$ polynomials. Let us denote $p_{C[i,j]}(x)$ as the $(i,j)$-th polynomial for $C[i,j]$ and let the row vector representation of the coefficients of $p_{C[i,j]}(x)$ as $\mbp_{[i,j]}$. 

\begin{lemma}\label{lem:p_norm}
For $p_\mbC(x)$ given in Construction~\ref{const:MatDot}, the norm of $\mbp_{[i,j]} = \mathsf{vec}(p_{C[i,j]})$ is bounded as:
\begin{equation}
    ||\mbp_{[i,j]}||_2 \leq \sqrt{2m-1} \eta^2,
\end{equation}
if $||\mbA||_F \leq \eta$ and $||\mbB||_F \leq \eta$. 
\end{lemma}
\begin{proof}
Throughout the proof, $||\cdot||$ denotes a Frobenius norm for a matrix and a 2-norm for a vector. Let $\mbP_l$ be the coefficient of $x^{l-1}$ in $p_{\mbC}(x)$ for $l \in [2m-1]$, which can be written as:
\begin{align*}
    \mbP_l &= \sum_{\substack{1\leq i,j \leq m \\ j-i=m-l}} \mbA_i \mbB_{j} = \begin{cases} 
    \sum_{1\leq i \leq l} \mbA_i \mbB_{i+m-l}, \text{ if } l \leq m\\ 
    \sum_{l+1-m\leq i \leq m} \mbA_i \mbB_{i+m-l}, \text{ otherwise.}
    \end{cases}
\end{align*}
Let us focus on the case when $l \leq m$ as the argument extends naturally for $l > m$. For $l \leq m$, $\mbP_l$ can be rewritten as: 
$
    \mbP_l = \sum_{1\leq i \leq l} \mbA_i \mbB_{i+m-l} = \begin{bmatrix}
     \mbA_1 & \cdots & \mbA_l
    \end{bmatrix} \cdot \begin{bmatrix}
     \mbB_{m-l+1} \\ \vdots \\ \mbB_{m}
    \end{bmatrix}.
$
As these matrices are submatrices of $\mbA$ and $\mbB$, 
\begin{equation}
    ||\begin{bmatrix}
     \mbA_1 & \cdots & \mbA_l
    \end{bmatrix} || \leq \eta , \;\; \norm{\begin{bmatrix}
     \mbB_{m-l+1} \\ \vdots \\ \mbB_{m}
    \end{bmatrix}} \leq \eta.
\end{equation}
Since $||\mbX \mbY|| \leq ||\mbX|| \cdot ||\mbY||$, we have $
|| \sum_{1\leq i \leq l} \mbA_i \mbB_{i+m-l} || \leq \eta^2.$
We can apply the same argument for $l >m$ and show: 
$
   ||\mbP_l|| = || \sum_{\substack{1\leq i,j \leq m \\ j-i=m-l}} \mbA_i \mbB_{j} || \leq \eta^2 \;\; \text{ for } l \in [2m-1]. 
$
Finally,
\begin{align}
    || \mbp_{[i,j]} || & = \norm{\begin{bmatrix}
     P_1[i,j] & P_2[i,j] & \cdots & P_{2m-1}[i,j]
    \end{bmatrix}} = \sqrt{\sum_{l=1}^{2m-1} P_l[i,j]^2} \leq \sqrt{\sum_{l=1}^{2m-1} ||\mbP_l||^2}  \\
    &\leq \sqrt{2m-1}\eta^2.
\end{align}
\end{proof}


\begin{alg}[Decoding of Approximate MatDot codes]
Let $ \mblambda^{(\text{succ})}$ be a length-$K$ vector with evaluation points at $K$ successful worker nodes:
$
    \mblambda^{(\text{succ})}= \begin{bmatrix}
     \lambda_{i_1}, \cdots, \lambda_{i_K}
    \end{bmatrix},
$
and let $\mbV^{(\text{succ})} = \mathsf{Vander}(\mblambda^{(\text{succ})}, 2m-2)$. Finally, we denote $\mby^{(\text{succ})}_{[i,j]}$ as the evaluations of $p_{C[i,j]}(x)$ at $\mblambda^{(\text{succ})}$, i.e., 
$
    \mby^{(\text{succ})}_{[i,j]} = \begin{bmatrix}
     \widetilde{C}_{i_1}[i,j] & \cdots \widetilde{C}_{i_K}[i,j]
    \end{bmatrix}.
$
For decoding $C[i,j]$, we solve the following optimization: 
\begin{equation}\label{eq:decoding_opt}
    \widehat{\mba} = \underset{\mba \mbV^{(\text{succ})} = \mby^{(\text{succ})}_{[i,j]}}{\mathrm{argmin}} \; \;
    ||\mba||_2. 
\end{equation}
If $||\widehat{\mba}||_2 > \sqrt{2m-1} \eta^2$, declare failure. Otherwise, 
$\widehat{C}[i,j] = \widehat{a}[m]$. 
\end{alg}

\textit{Proof of Theorem~\ref{thm:approx_matdot}:} 
First, note that the solution of the equation \eqref{eq:decoding_opt} and the true polynomial coefficients $\mbp_{[i,j]}$ both lie in $\calB_{\sqrt{2m-1} \eta^2}$. As  
$
    ||\lambda^{(\text{succ})}||_\infty < \frac{\epsilon}{6 \eta^2 \sqrt{2m-1}(m-1)m}, 
$
by construction, Corollary~\ref{cor:main} gives:
$
    |\widehat{a}[l] - p_{[i,j]}[l]| \leq \epsilon \;\; \text{ for } l \in [m].
$
Hence, 
$
    |\widehat{C}[i,j] - C[i,j]| = |\widehat{a}[m] - p_{[i,j]}[m]| \leq \epsilon. 
$

\section{} \label{app:conproof}
\begin{proof}[Proof of Theorem \ref{thm:con}]
	We show a contradiction, i.e., assume $K(m,\epsilon)=m-1,\;\forall\epsilon<\eta^2$.
	    We need to show that there exist matrices $\mbA,\mbB$ such that $\epsilon\geq \eta^2$ for a recovery threshold of $m-1$.
	
	Consider $f_i$ and $g_i$ defined in the system model \eqref{eq:fimapdef} and \eqref{eq:gimapdef}. Let
	$\bm{f}=\{f_i\}_{i=1}^P$ and $\bm{g}=\{g_i\}_{i=1}^P$ be encoding functions for $\mbA$ and $\mbB$ respectively. Consider any set $\mathcal{S}$ of $m-1$ nodes. Let $\bm{f}_{\mathcal{S}}$, $\bm{g}_{\mathcal{S}}$ denote the restriction of $\bm{f}$, $\bm{g}$ to the nodes corresponding to $\mathcal{S}$ respectively.	Let $\bm{C_i}=$ output of $i$\textsuperscript{th} node, $i\in \mathcal{S}$. 
		\[\bm{C_i} = (f_i(\{\mbA_{j,k}\}_{j=1,k=1}^{p,q}))\;(g_i(\{\mbB_{j,k}\}_{j=1,k=1}^{q,p}))\]
		Let $\bm{C}=\{\bm{C_i}\}_{i\in \mathcal{S}}$. Let $d_{\mathcal{S}}(\,\cdot\,;\bm{f}_{\mathcal{S}}, \bm{g}_{\mathcal{S}})$ denote any decoding function corresponding  to the $m-1$ nodes in $\mathcal{S}$ that takes $\bm{C}$ and gives an estimate of $\mbA\mbB$. To show a contradiction, we show that there exist matrices $\mbA,\mbB$, such that
		$||d_{\mathcal{S}}(\bm{C})-\mbA\mbB||_F\geq \eta^2 > 0 $
		Let $\textbf{vectorize}(\cdot)$ be a function that outputs a column-wise vectorization of the input matrix. Let $\bm{Q}\in\mathbb{R}^{p\times q}$ such that $\sigma_{\max}(\bm{Q})=1$ and, $\textbf{vectorize}(\bm{Q})$ is a null vector of $\bm{f}_{\mathcal{S}}$ i.e., $\bm{f}_{\mathcal{S}}(\textbf{vectorize}(\bm{Q})\otimes \bm{D})=\bm{0}$, $\forall\bm{D}\in\mathbb{R}^{\frac{n}{p}\times \frac{n}{q}}$. Note that we can scale any such $\bm{Q}$, such that its maximum singular value is 1.
		Let $\bar{\mbA}$ and $\bar{\mbB}$ be some constant matrices with bounded frobenius norms. We set 
		\[\mbA = \bm{Q}\otimes\bar{\mbA},\;\bar{\mbA}\in\mathbb{R}^{\frac{n}{p}\times \frac{n}{q}}\;\text{ and }\;
		\mbB^{(b)} = b\bm{Q}^T\otimes\bar{\mbB},\;\bar{\mbB}\in\mathbb{R}^{\frac{n}{q}\times \frac{n}{p}},\;b\in\mathbb{R}\text{\textbackslash}\{0\}\]
		Note: $\mbA\mbB^{(b)}=b\bar{\mbA}\bar{\mbB}$. Also observe that, $||\bar{\mbA}||_F\leq \eta$ and $||\bar{\mbB}||_F\leq \frac{\eta}{|b|}$.
		
		Let $\bm{C}_i^{(b)} = (f_i(\{\mbA_{j,k}\}_{j=1,k=1}^{p,q}))\;(g_i(\{\mbB^{(b)}_{j,k}\}_{j=1,k=1}^{q,p}))$, 
		$\bm{C}^{(b)}=\{\bm{C}_i^{(b)}\}_{i\in \mathcal{S}}$.
		
		By construction $\bm{C}^{(b)}=\bm{0}\implies d_{\mathcal{S}}(\bm{C}^{(1)})=d_{\mathcal{S}}(\bm{C^}{(-1)})=d_{\mathcal{S}}(\bm{0})$. Then by triangle inequality,
		\begin{align*}
			||d_{\mathcal{S}}(\bm{C}^{(1)}) - \mbA\mbB^{(1)}||_F + ||d_{\mathcal{S}}(\bm{C}^{(-1)})-\mbA\mbB^{(-1)}||_F &\geq ||\mbA\mbB^{(1)}-\mbA\mbB^{(-1)}||_F\geq 2||\mbA\mbB^{(1)}||_F
		\end{align*}
		\[\max(||d_{\mathcal{S}}(\bm{0}) - \mbA\mbB||_F , ||d_{\mathcal{S}}(\bm{0}) + \mbA\mbB||_F) \geq ||\mbA\mbB^{(1)}||_F\]
		\begin{align*}
		    ||\mbA\mbB^{(1)}||_F &= ||(\bm{Q}\otimes\bar{\mbA})(\bm{Q}^T\otimes\bar{\mbB})||_F 
		    = ||(\bm{Q}\bm{Q}^T)\otimes(\bar{\mbA}\bar{\mbB})||_F \\
		    &= \sqrt{\Tr((\bm{Q}\bm{Q}^T \otimes \bar{\mbB}^T\bar{\mbA}^T)(\bm{Q}\bm{Q}^T \otimes \bar{\mbA}\bar{\mbB}))} 
		    = \sqrt{\Tr((\bm{Q}\bm{Q}^T \bm{Q}\bm{Q}^T )\otimes (\bar{\mbB}^T\bar{\mbA}^T \bar{\mbA}\bar{\mbB}))}\\
		    &= \sqrt{\Tr(\bm{Q}\bm{Q}^T \bm{Q}\bm{Q}^T )}\sqrt{\Tr(\bar{\mbB}^T\bar{\mbA}^T \bar{\mbA}\bar{\mbB})}\\
		    &=||\bm{Q}\bm{Q}^T||_F||\bar{\mbA}\bar{\mbB}||_F
		\end{align*}
		We can find matrices $\bar{\mbA},\bar{\mbB}$ such that $||\bar{\mbA}\bar{\mbB}||_F=\eta^2$ due to Lemma \ref{lem:con}.
		\begin{align*}
		    ||\mbA\mbB^{(1)}||_F &= ||\bm{Q}\bm{Q}^T||_F\,
		    \eta^2 \geq ||\bm{Q}\bm{Q}^T||_2\,\eta^2
		    =||\bm{Q}||_2^2\,\eta^2
		    = \eta^2
		\end{align*}
		 Therefore 
	    \[\max(||d_{\mathcal{S}}(\bm{0}) - \mbA\mbB||_F , ||d_{\mathcal{S}}(\bm{0}) + \mbA\mbB||_F) \geq \eta^2\]

	Therefore, there exists matrices $\mbA$ and ($\mbB^{(1)}$ or $\mbB^{(-1)}$) such that given $\eta=1$, the decoding error is $\geq  1$, when recovery threshold is set to $m-1$.
	\end{proof}
		
	\begin{lemma}
	\label{lem:con}
	    Choose $\bar{\mbA}=\bm{x}\bm{y}^T$ and $\bar{\mbB}=\bm{y}\bm{z}^T$, then $||\bar{\mbA}\bar{\mbB}||_F=||\bar{\mbA}||_F||\bar{\mbB}||_F$.
	\end{lemma}
	\begin{proof}
	\begin{align*}
	    ||\bar{\mbA}\bar{\mbB}||_F &= ||\bm{y}||^2 ||\bm{x}\bm{z}^T||_F
	    =||\bm{y}||^2 \sqrt{\Tr{(\bm{z}\bm{x}^T\bm{x}\bm{z}^T)}}
	    = ||\bm{y}||^2 ||\bm{x}||||\bm{z}||\\
	    ||\bar{\mbA}||_F||\bar{\mbB}||_F &= \sqrt{\Tr{(\bm{y}\bm{x}^T\bm{x}\bm{y}^T)}} \sqrt{\Tr{(\bm{z}\bm{y}^T\bm{y}\bm{z}^T)}}
	    = ||\bm{y}||^2 ||\bm{x}||||\bm{z}||
	\end{align*}	\end{proof}


\section{Proof of Theorem~\ref{thm:approx_polydot}}
We first prove the following  crucial theorem.
\begin{theorem}\label{thm:alt_prf}
Let $f(x) = a_0 + a_1 x + \cdot + a_{k-1} x^{k-1}$ and let $x_1, \ldots, x_m$ be distinct real numbers that satisfy $|x_i| \leq \delta$ for all $i = 1, \ldots, m$, for some $0 < \delta < \frac{1}{m}$. Let
\begin{equation}
    \begin{bmatrix}
    \widehat{a}_0 \\ \vdots \\  \widehat{a}_{m-1}
    \end{bmatrix} \triangleq \mbV^{-1} \begin{bmatrix}
    f(x_1) \\ \vdots \\ f(x_m)
    \end{bmatrix},
\end{equation}
where $\mbV =\vander{ \mbx, m}$ for $\mbx = [x_1 \;\; \cdots \;\; x_m ] $. Then, 
\begin{equation}
    | \widehat{a}_{m-1} - a_{m-1} | \leq ||\mba||_\infty \cdot (k-m)m\delta.
\end{equation}
\end{theorem}
\begin{proof}
Let $R_m(x)$ be the higher order terms in f: 
$
    R_m(x) = a_m x^m + \cdots a_{k-1} x^{k-1}.
$
Then, the following relation holds: 
\begin{equation*}
    \begin{bmatrix}
    f(x_1) \\ \vdots \\ f(x_m)
    \end{bmatrix} = \mbV \begin{bmatrix}
    a_0 \\ \vdots \\ a_{m-1}
    \end{bmatrix} + \begin{bmatrix}
    R_m(x_1) \\ \vdots \\ R_m(x_m)
    \end{bmatrix}
    \;\;
    \iff 
    \;\;
    \mbV^{-1} \begin{bmatrix}
    f(x_1) \\ \vdots \\ f(x_m)
    \end{bmatrix} = \begin{bmatrix}
    a_0 \\ \vdots \\ a_{m-1}
    \end{bmatrix} + \mbV^{-1} \begin{bmatrix}
    R_m(x_1) \\ \vdots \\ R_m(x_m)
    \end{bmatrix}.
\end{equation*}
Let $\mbv$ be the last row of $\mbV^{-1}$, i.e., $\mbv = \mbV^{-1}[m,:]$ and let $\mbr = [ R_m(x_i)]_{i\in [m]}$. Then,
$
    | \widehat{a}_{m-1} - a_{m-1} | = | \mbv \cdot \mbr |. 
$
Using the explicit formula for the inverse of Vandermonde matrices, the $i$-th entry of $\mbv$ is given as:
\begin{equation}
    v[i] = \frac{1}{\prod_{\substack{j = 1 \\ j \neq i}}^{m} (x_j - x_i)}.
\end{equation}
Thus, $\mbv \cdot \mbr$ can be rewritten as:
\begin{align}
    \mbv \cdot \mbr = \sum_{i=1}^{m} \frac{\sum_{l=1}^{k-m} a_{m-1+l} x_i^{m-1+l}}{\prod_{\substack{j = 1 \\ j \neq i}}^{m} (x_j - x_i)} = \sum_{l=1}^{k-m} a_{m-1+l}  \sum_{i=1}^{m} \frac{x_i^{m-1+l}}{\prod_{\substack{j = 1 \\ j \neq i}}^{m} (x_j - x_i)}. \label{eq:alt_prf1}
\end{align}

The expression in \eqref{eq:alt_prf1} can be further simplified using the following lemma. 
\begin{lemma} \label{lem:alt_prf_lem} [Theorem 3.2 in \cite{cornelius2011identities}]
\begin{equation}
     \sum_{i=1}^{m} \frac{x_i^{m-1+l}}{\prod_{\substack{j = 1 \\ j \neq i}}^{m} (x_j - x_i)} = h_l (x_1, \ldots, x_m),
\end{equation}
where $h_l$ is the complete homogeneous symmetric polynomial of degree $l$ defined as:
\begin{equation}
    h_l(x_1, \ldots, x_m) = \sum_{d_1 + \cdots + d_m = l} x_1^{d_1} \cdot x_2^{d_2}  \cdots  x_m^{d_m}.
\end{equation}
\end{lemma}

Using Lemma~\ref{lem:alt_prf_lem}, \eqref{eq:alt_prf1} can now be written as: 
$
     \mbv \cdot \mbr  = \sum_{l=1}^{k-m} a_{m-1+l} \cdot  h_l(x_1, \cdots, x_m). 
$
Finally, 
\begin{align}
     | \widehat{a}_{m-1} - a_{m-1} | &= | \mbv \cdot \mbr | = \left| \sum_{l=1}^{k-m} a_{m-1+l} \cdot  h_l(x_1, \cdots, x_m) \right|  \nonumber \\ 
     &\leq \sum_{l=1}^{k-m} a_{m-1+l} \cdot \binom{m+l-1}{l} \delta^l \leq ||\mba||_\infty \sum_{l=1}^{k-m} \binom{m+l-1}{l} \delta^l \nonumber \\ 
     &\leq ||\mba||_\infty (k-m) m \delta. 
\end{align}
The last inequality holds because $\delta < \frac{1}{m}$ and thus $\binom{m+l-1}{l} \delta^{l} \leq \binom{m}{1} \delta $ for $l = 1, \ldots, m-1$.
\end{proof}

Recall that  polynomial $p_{\mbC}(x)$ is essentially a set of $n^2$ polynomials, having one polynomial for each $C[i,j]$ ($i,j \in [n]$, and we use $p_{C[i,j]}(x)$ as the $(i,j)$-th polynomial for $C[i,j]$.

\begin{lemma}\label{lem:polydot_norm}
Assume $||\mbA||_F \leq \eta$ and $||\mbB||_F \leq \eta$. Then, for $p_\mbC(x)$ given in Construction~\ref{const:polydot}, the $\infty$-norm of $\mbp_{[v,w]} = \mathsf{vec}(p_{C[v,w]})$  ($v,w \in [n]$) is bounded as:
\begin{equation}
    ||\mbp_{[v,w]}||_\infty \leq \eta^2.
\end{equation}
\end{lemma}
\begin{proof}
Let $d \triangleq q(i-1) + pq(l-1)+ (q-1+j-k)$. The coefficient of $y^{d}$ in $p_\mbC(x)$ is:
\begin{equation}
\mbP_{d} =\begin{cases}
     \sum_{j'-k' = j-k} \mbA_{i,j'} \mbB_{k', l} + \sum_{j'-k' = j-k-q} \mbA_{i+1,j'} \mbB_{k', l}, \;\; &\text {for } j-k > 0, \\
     \sum_{j'=k'} \mbA_{i,j'} \mbB_{k', l},  \;\; &\text {for }  j-k = 0.
\end{cases}
\end{equation}
For both cases, the number of terms in the sum is $q$. Thus, it can be rewritten as:
\begin{equation}
\begin{bmatrix}
     \mbA_{j_1} & \cdots & \mbA_{j_q}
    \end{bmatrix} \cdot \begin{bmatrix}
     \mbB_{k_1} \\ \vdots \\ \mbB_{k_q}
    \end{bmatrix}.
\end{equation}
As these matrices are submatrices of $\mbA$ and $\mbB$, 
\begin{equation}
    \norm{\begin{bmatrix}
     \mbA_{j_1} & \cdots & \mbA_{j_q}
    \end{bmatrix}}_2 \leq \eta , \;\; \norm{\begin{bmatrix}
     \mbB_{k_1} \\ \vdots \\ \mbB_{k_q}
    \end{bmatrix}}_2 \leq \eta,
\end{equation}
Hence,
$~~
    ||\mbp_{[v,w]}||_\infty = \max_{d} |P_d[v,w]| \leq ||\mbP_d||_2 \leq \eta^2.
$\end{proof}

\begin{proof}[Proof of Theorem~\ref{thm:approx_polydot}]
The decoding for $\epsilon$-approximate PolyDot codes can be performed as follows. For decoding $\mbC_{i,l}$, we choose $d_{i,l} = iq+pq(l-1)$ points from the $p^2q$ successful nodes. Let $\mbV_{i,l} = \vander{[x_1, \cdots, x_{d_{i,l}}], d_{i,l}}$ and $\mbv$ be the last row of $\mbV_{i,l}^{-1}$. Then, we decode  $\mbC_{i,l}$ by computing:
\begin{equation}
    \widehat{\mbC}_{i,l} = \mbv \cdot \begin{bmatrix}
         p_\mbC(x_1) \\ \vdots \\ p_\mbC(x_{d_{i,l}})
    \end{bmatrix}.
\end{equation}
By combining Theorem~\ref{thm:alt_prf} and Lemma~\ref{lem:polydot_norm}, we can show that: 
\begin{align*}
    \norm{\widehat{\mbC}_{i,l} - \mbC_{i,l}}_{\max} \leq \frac{(p^2q+q-1 - d_{i,l})d_{i,l}}{q(p^2q-1)} \epsilon.
\end{align*}
The smallest $d_{i,l}$ is $d_{1,1} = q$ and the largest $d_{i,l}$ is $d_{p,p} = p^2q$. For $q \leq d_{i,l} \leq p^2q$, $(p^2q+q-1 - d_{i,l})d_{i,l} \leq q(p^2q-1)$. Hence, 
$
    \norm{\widehat{\mbC}_{i,l} - \mbC_{i,l}}_{\max} \leq \epsilon.
$
\end{proof}

\section{Proof of Theorem~\ref{thm:opt_bound}} \label{app:optboundproof}
\begin{proof}
Let $\mbE =  \mbI_{m \times m} - \sum_{i \in \calS_p} d_i^{(p)} \mbalpha^{(i)} \mbbeta^{(i) T}$. $||\cdot||$ here represent Frobenius norm. 
\begin{align*}
    ||\mbC - \widehat{\mbC}_{\calS_p}|| &= ||\sum_{i,j} E[i,j] \mbA_i \mbB_j|| \leq \sum_{i,j} |E[i,j]| \cdot ||\mbA_i \mbB_j|| \leq \sum_{i,j} |E[i,j]| \cdot ||\mbA \mbB|| \\
    &\leq m ||\mbE|| \cdot ||\mbA \mbB|| \leq m ||\mbE|| \cdot ||\mbA || \cdot || \mbB|| = m \sqrt{\ell^{(p)}}  \eta^2.
\end{align*}
\end{proof}

\end{document}